\documentclass[twocolumn,10pt]{IEEEtran}

\usepackage{calc}

\usepackage{xcolor,multirow,gensymb}
  \usepackage{cite}
\usepackage{subfigure,epsfig}
\usepackage{graphicx}
\usepackage{caption}

\usepackage{psfrag}
\usepackage{amsmath,amssymb}
\usepackage{color}
\usepackage{array}
\usepackage{nicefrac}
\usepackage{amsmath,amsthm,amssymb}
\usepackage{xifthen}
\usepackage{mathtools}
\usepackage{enumerate}
\usepackage{microtype}
\usepackage{xspace}
\usepackage{bm}
\usepackage[T1]{fontenc}
\usepackage{fancyhdr}
\usepackage{lastpage}
\usepackage{rotating}
\usepackage{tabulary}

\newcommand{\mbs}[1]{\bm{#1}}
\newcommand{\mat}[1]{{\uppercase{\mbs{#1}}}}
\newcommand{\Id}{\mat{\mathrm{I}}}

\newcommand{\T}{{\scriptscriptstyle\mathsf{T}}}
\renewcommand{\H}{{\scriptscriptstyle\mathsf{H}}}

\renewcommand{\Re}[1][]{\ifthenelse{\isempty{#1}}{\operatorname{Re}}{\operatorname{Re}\left(#1\right)}}
\renewcommand{\Im}[1][]{\ifthenelse{\isempty{#1}}{\operatorname{Im}}{\operatorname{Im}\left(#1\right)}}
\usepackage{float}

\newfloat{longequation}{b}{ext}

\def\bA{{\mathbf{A}}}

\def\bI{{\mathbf{I}}}

\def\bN{{\mathbf{N}}}

\newcommand{\cC}{{\cal C}}

\newcommand{\cN}{{\cal N}}

\def\bb{{\mathbf{b}}}

\def\bee{{\mathbf{e}}}

\def\bh{{\mathbf{h}}}

\def\bn{{\mathbf{n}}}

\def\bq{{\mathbf{q}}}

\def\bs{{\mathbf{s}}}

\def\bw{{\mathbf{w}}}
\def\bx{{\mathbf{x}}}
\def\by{{\mathbf{y}}}

\def\b0{{\mathbf{0}}}

\def\bbC{{\mathbb{C}}}

\def\bpsi{{\boldsymbol{\psi}}}

\newcommand{\EE}{\mathbb{E}}

\newcommand{\CN}[1][]{\ifthenelse{\isempty{#1}}{\mathcal{N}_{\mathbb{C}}}{\mathcal{N}_{\mathbb{C}}\left(#1\right)}}

\renewcommand{\P}[1][]{\ifthenelse{\isempty{#1}}{\mathbb{P}}{\mathbb{P}\left(#1\right)}}
\renewcommand{\det}[1][]{\ifthenelse{\isempty{#1}}{\text{det}}{\text{det}\left(#1\right)}}
\newcommand{\trace}[1][]{\ifthenelse{\isempty{#1}}{\text{tr}}{\text{tr}\left(#1\right)}}
\newcommand{\rank}[1][]{\ifthenelse{\isempty{#1}}{\text{rank}}{\text{rank}\left(#1\right)}}
\newcommand{\diag}[1][]{\ifthenelse{\isempty{#1}}{\text{diag}}{\text{diag}\left(#1\right)}}
\def\nn{\nonumber}

\IEEEoverridecommandlockouts

\theoremstyle{plain}
\newtheorem{proposition}{Proposition}
\newtheorem{theorem}{Theorem}
\newtheorem{lemma}{Lemma}

\def\bPhi{{\boldsymbol{\Phi}}}

\newcommand{\tr}{\mathop{\mathrm{tr}}\nolimits}

\newcommand{\al}{\alpha}

\newcounter{enumi_saved}
\usepackage{answers}
\Newassociation{solution}{Solution}{solutionfile}

\AtBeginDocument{\Opensolutionfile{solutionfile}[\jobname]}
\AtEndDocument{\Closesolutionfile{solutionfile}\clearpage
}

\usepackage[utf8]{inputenc}
\usepackage[english]{babel}
\usepackage{dsfont}
\usepackage{calc,cool}
\usepackage{minibox}
\DeclareSymbolFont{matha}{OML}{txmi}{m}{it}
\DeclareMathSymbol{\varv}{\mathord}{matha}{118}
\usepackage{mathtools}
\usepackage{slashed}
\usepackage{eurosym}

\usepackage{amssymb}

\makeatletter

\IEEEoverridecommandlockouts
\begin{document}
\title{Modeling and Performance of Uplink Cache-Enabled Massive MIMO Heterogeneous Networks}
\author{Anastasios Papazafeiropoulos and Tharmalingam Ratnarajah   \vspace{2mm} \\
\thanks{A. Papazafeiropoulos was with the  Institute for Digital Communications (IDCOM), University of Edinburgh, Edinburgh, EH9 3JL, U.K. He is now with  the Communications and Intelligent Systems Group,
University of Hertfordshire, AL10 9AB Hatfield , U. K. and with SnT at the University of Luxembourg, L-1855 Luxembourg. T. Ratnarajah is  with IDCOM, University of Edinburgh, Edinburgh.
Email: tapapazaf@gmail.com; t.ratnarajah@ed.ac.uk}
\thanks{This work was supported by the U.K. Engineering and Physical Sciences Research Council (EPSRC) under Grant EP/N014073/1 and the UK-India Education and Research Initiative Thematic Partnerships under grant number DST-UKIERI-2016-17-0060.}}\maketitle
%


\vspace{-1.5cm}

\begin{abstract}
A significant burden on wireless networks is brought by the uploading of user-generated contents to the Internet by means of applications such as the social media. To cope with this mobile data tsunami, we develop a novel multiple-input multiple-output (MIMO) network architecture with randomly located  base stations (BSs) a large number of antennas employing cache-enabled \textit{uplink} transmission. In particular,  we formulate a scenario, where the users upload their content to their strongest  base stations (BSs), which are Poisson point process (PPP) distributed. In addition, the BSs, exploiting the benefits of massive MIMO,  upload  their contents to the core network by means of a finite-rate backhaul.  After proposing the caching policies, where we propose the {modified} von Mises distribution as the popularity distribution function, we derive the outage probability and the average delivery rate by taking  advantage  of   tools from the deterministic equivalent (DE) and stochastic geometry analyses. Numerical results investigate the realistic performance gains of the proposed  heterogeneous cache-enabled uplink  on the network in terms of cardinal operating parameters. For example, insights regarding the BSs storage size are exposed. Moreover, the impacts of the key parameters such  the file popularity distribution, and the target bitrate are investigated. Specifically, the outage probability decreases if the storage size is increased, while the average delivery rate increases. In addition, the concentration parameter, defining the number of files stored at the  intermediate nodes (popularity), affects directly the proposed metrics. Furthermore, a higher target rate results in higher outage because fewer users obey this constraint. Also, we demonstrate that a denser network decreases the outage and increases the delivery rate. Hence, the introduction of caching at the uplink of the system design ameliorates the network performance.
\end{abstract}
\begin{keywords}
Caching, channel aging,  heterogeneous networks, massive MIMO,  stochastic geometry
\end{keywords}
\section{Introduction}
A vast majority of new wireless services such as social networks, web-browsing applications, and multimedia streaming has fueled the mobile data traffic with an imminent $500$-fold boost over the next 10 years~\cite{Index}. As a result, mobile operators need to redesign their current networks and delve into more sophisticated techniques to surge system capacity forward and expand coverage in fifth generation (5G) networks~\cite{Andrews2014}.

A promising solution towards this direction relies on the deployment of low-power,  short-range, and cost-efficient small cell networks (SCNs) or else heterogeneous networks (HetNets). In fact, irregular cellular networks, deployed opportunistically and in hot spots have been researched for a fairly long time now~\cite{Andrews2012}. Specifically, downlink single-input single-output (SISO) HetNets have already presented sufficient progress~\cite{Andrews2011,Madhusudhanan2011,Jo2012,Mukherjee2012}. Literally, relevant standardization activities have started in 3GPP release a long time ago~\cite{Sankaran2012}. Having assumed that SISO studies are anachronistic, efforts have been devoted to the challenge of modeling multi-antenna HetNets by assuming perfect channel state information (CSI)~\cite{Heath2013,Dhillon2013,PapazafComLetter2016}. In addition, the practical consideration of imperfect CSI has taken place in several works~\cite{Kountouris2012,Papazafeiropoulos2017}. Moreover, research has been devoted to the analysis of stochastically geometric uplink models~\cite{Novlan2013,Bai2016}.

In a parallel avenue, massive multiple-input multiple-output (MIMO) has emerged as another technology supporting the backbone of 5G networks. Remarkably, massive MIMO point to the increase of spectral and energy efficiencies~\cite{Ngo2013,Lu2014}. The achievement of high cell-throughput along with simple signal processing has been contrived by deploying large-scale antenna arrays at the BS and multi-user (MU) transmission. Starting from the strong assumption of perfect CSI, research has faced realistic impediments such as the presence of pilot contamination~\cite{Marzetta2010} and~\cite{Ngo2013}, the inevitable hardware impairments~{\cite{Bjornson2014} and ~{\cite{Bjornson2015}, as well as the channel aging~\cite{Papazafeiropoulos2015a,Papazafeiropoulos2015,Papazafeiropoulos2016}. Especially, by applying the theory of deterministic equivalent (DE) analysis, massive MIMO systems were studied under the presence of channel aging~\cite{Papazafeiropoulos2015a}. The key assumption of the DE analysis is that $M \to \infty$ and $K \to \infty$ with a given ratio, where $M$ and $K$ are the numbers of BS antennas and users, respectively. In particular, channel aging refers to the channel variation between the time instance the channel is estimated and the time instance it is used for precoding or detection. The sources of channel aging are mainly the relative movement of the users with the BS antennas, the phase noise, and any processing delays~\cite{Papazafeiropoulos2016}. Despite its significant implications, few works have scrutinized its impact on massive MIMO systems~\cite{Papazafeiropoulos2015a,Papazafeiropoulos2015,Papazafeiropoulos2016,Papazafeiropoulos2017}.

The growing trend of user-generated content  such as the sharing of real-time events by means of smartphones inflicts a great uploading strain to the wireless networks. Despite the importance of uplink on mobile cellular networks, most efforts have been focused on the downlink scenario~\cite{Bastug2014,Liu2014,Zhao2016}. In fact, few attempts have been dedicated to the expanding demands of users' transmission (uploading). Notably, differences appear between the procedures of uploading and downloading. Specifically, an asymmetry regarding the downlink and uplink bandwidths takes place, since the downlink bandwidth can reach $10-1000$ times the corresponding uplink bandwidth. As a consequence, the uploading time is longer,  and the throughput is lower. Hence, lower quality of experience is met. Another difference concerns the limited resources of mobile devices such as the battery capacity and the transmit power. Obviously, it is a critical solution to moderate the uplink traffic pressure. An efficient remedy that is brought to the play is the employment of caching in SCNs by exploiting the content popularity appearing as redundancy~\cite{Bastug2015}. In fact, caching the users' content in the edge of the network brings gains regarding the user  satisfaction and the traffic load~\cite{Bastug2014}. However, most works address the problem of caching in the downlink direction.

\subsection{Motivation-Central Idea}
This paper is motivated by the following observations: 1) Content providers relocate their users' contents from the core network to the intermediate nodes in the network (caching). Hence, a key question to answer is the design and investigation of the converse scenario, where the users move their content to the intermediate nodes to alleviate the upload traffic. 2) We employ a large number of antennas at each BS to take advantage of the benefits of massive MIMO such as the elimination of intra-cell interference. 3) Networks have the tendency to become denser resulting in their irregularity, i.e., it is required to introduce the concept of SCNs. 4) User mobility and its resultant channel aging is a common phenomenon in wireless communications.
The main contributions are summarized as follows.
\begin{itemize}
 \item Contrary to existing works~\cite{Bastug2014,Bastug2015}, which have studied the downlink caching, we focus on a novel strategy described as \textit{uplink caching}. Hence, instead of having the users requesting contents from their associated BS, the users upload their contents to their strongest BS. More concretely, we formulate the caching problem in the uplink. 
 \item  In parallel of considering uplink caching, we take advantage of the gained performance benefits when the massive MIMO technology  is employed and HetNet design is encountered. As far as the authors are aware, it is worthwhile to mention that this work is  unique regarding the study of the notion of massive MIMO in caching.
 \item We introduce the modified von Mises distribution instead of a power-law or the Zipf distribution to describe  the content popularity distribution, in order to represent the locality and the concentration of  the content popularities in a specific region.
 \item It is the first work introducing channel aging in an architecture including caching. Although we do not  show the degradation of the system  performance in terms of plots by varying the user mobility due to users' relative movement, the analytical results allow the observation of its dependence and its loss quantification.
  \item We derive  the outage probability and the average delivery rate of an uplink massive MIMO HetNet, where the intermediate nodes are enriched with caching resources. {In particular, after having obtained the deterministic signal-to-interference-plus-noise ratio (SINR) by means of the theory of DEs, we achieve to obtain a statistical expression}. Notably, the main benefit of the DEs is the provision of deterministic expressions allowing to avoid any Monte Carlo simulations.
    \item Relied on the numerical results, we elaborate on the  impact of various parameters such as the  BSs density and the storage size. For example, a high storage size induces improvement of the system, since the outage probability decreases and the average delivery rate increases. 
    For the sake of comparison, we also present the results corresponding to the absence of caching, where applicable.
  \end{itemize}

\subsection{Paper Outline}  
The remainder of this paper has the following structure.  Section~\ref{sec:systemmodel} develops the system model of the uplink of a massive MIMO HetNet with channel aging. Section~\ref{sec:caching} presents the caching model, while Section~\ref{sec:estimation} provides the estimated channel including the effects of pilot contamination and channel aging. Next, Section~\ref{sec:uplink} presents the uplink  transmission under the presence of channel aging and introduction of the caching concept. Section~\ref{sec:performance} provides the main results of this study. Especially, Subsection~\ref{outage15} includes the derivation and investigation of the outage probability, while Subsection~\ref{AchievableRate}, provides the presentation of the  average delivery rate of this general model.  The numerical results are placed in Section~\ref{Numerical}, and Section~\ref{Conclusion} concludes the paper.

\subsection{Notation}Vectors and matrices are denoted by boldface lower and upper case symbols. The notations $\mathcal{C}^{M \times 1}$ and $\mathcal{C}^{M\times N}$ refer to complex $M$-dimensional vectors and  $M\times N$ matrices, respectively. The symbols $(\cdot)^\T$, $(\cdot)^\H$, and $\tr\!\left( {\cdot} \right)$ express the transpose,  Hermitian  transpose, and trace operators, respectively. The expectation  operator is denoted by $\EE\left[\cdot\right]$, and the symbol $\triangleq$ declares definition.  Moreover, $\mathrm{J}_{0}(\cdot)$ is the zeroth-order Bessel function of the first kind, and $\Gamma\left( x,y \right)$ denotes the Gamma distribution with shape and scale parameters $x$ and $y$, respectively.  Finally, $\bb \sim \cC\cN{(\b0,\mathbf{\Sigma})}$ represents a circularly symmetric complex Gaussian vector with zero-mean and covariance matrix $\mathbf{\Sigma}$.
\section{System Model}\label{sec:systemmodel} 

\begin{figure*}[!h]
  \begin{center}
 \includegraphics[width=0.9\linewidth]{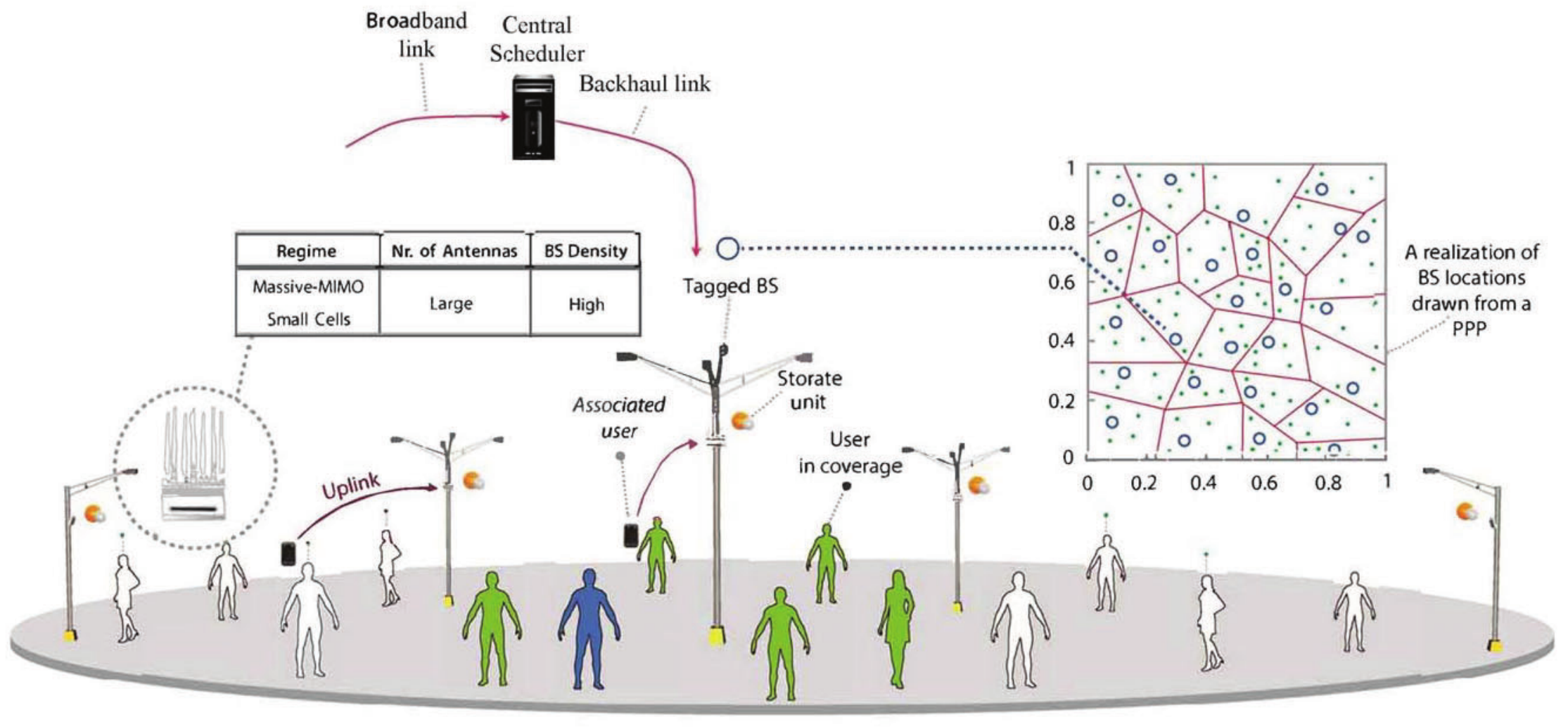}
 \caption{{An illustration of the massive MIMO heterogeneous network model with caching. The main figure shows the configuration of BSs, users, storage units, the backhaul, and their interconnection. The rectangular at the top right side illustrates a possible HetNet. The green solid disks and the grey rings represent the users and the BSs, respectively.}}
 \label{fig:scenario}
 \end{center}
 \end{figure*}

This section considers the setup for the uplink of a massive MIMO  cellular network consisted of BSs with locations drawn according to an independent PPP $\Phi_{\mathrm{B}}$ with density $\lambda_{\mathrm{B}}$, i.e., this formulation corresponds to the generalized and quite interesting design of uplink massive MIMO HetNets. For the sake of  better description, a graphical representation of the system layout is shown in Fig~\ref{fig:scenario}. Specifically, 
let the BS of the $l$th cell having a large number of antennas, denoted by $M_{l}$.  The users are assumed to be distributed as an independent PPP with a sufficiently high density $\lambda_{\mathrm{b}}$. In fact, in any resource block, we assume that  the $l$th BS randomly schedules $K_{l}$  users  according to a distance-based
criterion.  Actually, the $K_l$ users are connected with the nearest BS constituting its Voronoi cell, while a Voronoi tessellation is structured by the set of all these cells\cite{Bjoernson2016}\footnote{In this work, we assume  only a non-line-of-sight (NLoS) transmission, while the consideration of an LoS component is left for future work. Its introduction in the analysis could be made by means of a multi-slope path-loss model~\cite{Tang2001}.}.
In other words, these users are connected with the strongest BS constituting its Voronoi cell, while a Voronoi tessellation is structured by the set of all these cells\footnote{The users are assumed to be distributed as an independent point process, but each cell is large enough (the density of the users' PPP is sufficiently large) to shelter $K_{l}$ users. The users in each cell are independently and uniformly distributed.}. Also, we assume that $M_{l}\gg K_{l}$, as stated by the basic principle of massive MIMO technology\footnote{Although, in practice the number of BS antennas $M_{l}$, and the number of associated users  $K_{l}$ differ across cells, henceforth, we assume $M_{l}=M$ and $K_{l}=K$ for the sake of simplicity.}.  Moreover, we consider that each user is equipped with just a single-antenna mobile terminal. Evidently, since the massive MIMO concept is employed, many degrees of freedom are shared across each cell. A central scheduler provides a fixed broadband connection to these BSs by means of wired backhaul links. Obviously, the capacity of the link between the backhaul and each BS is a decreasing function of $\lambda_\mathrm{B}$, since the deployment of more BSs per given area results in less capacity per backhaul link. A capacity expression, obeying to this property, will be introduced in Subsection~\ref{AchievableRate}.

Further to the network topology, by exploiting Slivnyak's theorem, we are able to conduct the analysis just by focusing on a randomly chosen BS found at the origin~\cite{Chiu2013a}. Hereafter, we refer to this BS as the tagged BS. Hence, we assume  that at time $n$, the location of the associated scheduled user  is at $x_{llk,n}$, while the location of the $k$th user found in the $j$th cell  is denoted by $x_{ljk,n}$. Similarly,  $\bh_{ljk,n}\in \mathbb{C}^{M\times 1}$ is the  channel vector from the associated $k$th user in the  cell  located at $x_{ljk,n}$ to the tagged BS (located at the origin), while the interference term $\bh_{ljk,n}\in \mathbb{C}^{M\times 1}$ is the channel vector corresponding to the link from the $k$th user of the $j$th cell located at $x_{ljk,n}\in \mathbb{R}^{2}$. The locations of the $k$th scheduled users from all the cells are formed by a non-stationary point process $\Phi_{k}$, which is not a PPP because of the correlation of their locations  with the BS process. The explanation relies on the prohibition of  the presence of all other users in
$\Phi_{k}$ in the tagged cell~\cite{Novlan2013,Bai2016}. Although this kind of correlations regarding the scheduled users' locations  make the exact
analysis intractable, we endorse the uplink
model, accounting for the pairwise correlations,  proposed in~\cite{Liang2015} and followed in~\cite{Bai2016}. Moreover, we consider an exclusion ball approximation on the distribution of the scheduled user process $\Phi_{k}$, being a first-order approximation of the model in~\cite{Liang2015}.  On the top of the ball approximation, we assume that the random variable, expressing the  distance $\|x_{llk,n}\|^{-\al} $ from the scheduled user to its tagged  BS at the origin, is assumed to be a Rayleigh variable with a mean of $0.5\sqrt{1/\lambda_\mathrm{B}}$~\cite{Novlan2013}. We denote $R_{e}=\sqrt{1/\left( \pi \lambda_\mathrm{B} \right)}$ the radius of the ball, in order to let the size of the  surface of the exclusion ball equal the average cell size, which is $1/\lambda_\mathrm{B}$~\cite{Francois2009}. Furthermore, the scheduled user process $\Phi_{k}$, describing the locations of the other scheduled users, is formed as an inhomogeneous PPP of density $\lambda_{k}$ where the users are found outside an exclusion
ball having as center the tagged BS.  Especially, the locations of the $k$th users of the other cells, belonging to $\Phi_{k}$, are modeled by means of an inhomogeneous PPP with a density function of 
\begin{align}
~ {\lambda_{k}\left( r \right)=\lambda_\mathrm{B}\left( 1-\mathrm{e}^{-\lambda_\mathrm{B} \pi r^{2}} \right)\!,}
\end{align}
where $r=\|x_{ljk,n}\|$.

Both the uplink training and data transmission phases necessitate the introduction  of fractional power control in our analysis. Thus, the $k$ user in the the $l$th cell transmits with power
\begin{align}
 P_{lk,n}=P_{\mathrm{t}}\beta^{-\epsilon}_{llk,n},
\end{align}
where $\epsilon \in \left[ 0,1\right] $ expresses  the fraction of the  compensation of the path-loss given by $\beta^{-\epsilon}_{llk,n}$, while $P_{\mathrm{t}}$ is the open loop transmit power assuming no power control.

As far as the channel model is concerned, let the point-to-point channels  be characterized by independent and identically distributed (i.i.d.) Rayleigh block fading with unit mean, while we assume a block fading model, where the channel is assumed constant during one block, but varies independently from block to block. Note that although the assumption of both line and non-line of sight signals appear in small cells, we consider only Rayleigh fading for the sake of simplicity. Relaxation of the Rayleigh fading assumption as well as the introduction and study of other fading models can be done with techniques found in~\cite{ElSawy2013}, and is left for future work.  Hence, in the proposed model, the channel vector $\bh_{ljk,n}$ from user $k$ in the $j$th cell  at the $n$-th time slot is modelled as
\begin{align}\label{eq:channelModel}
\bh_{ljk,n} = \beta^{1/2}_{ljk}\bw_{ljk,n},
\end{align}
where $\beta_{ljk}$ is  the large-scale path-loss and $\bw_{ljk,n} \in \bbC^{M\times 1}$ is an uncorrelated fast fading Gaussian channel vector with elements having zero mean and unit variance, i.e., $\bw_{ljk} \sim \cC\cN(\b0,\bI_{M})$. Note that the incorporation of spatial correlation due to lack of limited antenna spacing, and different antenna patterns is left for future work. The path-loss at the tagged cell is described by
\begin{align}
 \beta_{ljk}=C r_{ljk}^{-\al}, 
\end{align}
where $\al > 2$ is the path-loss exponent, and $C$ expresses a constant determined by the carrier frequency and the reference distance. Given that we have assumed an NLoS component, the distance-based criterion is translated to path-loss based, i.e., the minimum distance corresponds to the minimum path loss signal.
 For the sake of exposition,  we have assumed a simplistic single-slope path-loss model, in order not to distract the reader from the main contributions. The application of a more complex model, such as the multi-slope path-loss model presented in~\cite{Zhang2015}, is outside of the scope of this paper and left to future work. 


The transmission scheme includes an uplink channel estimation phase, and continues with an uplink data transmission phase, allowing the derivation of the outage probability and the average delivery rate, in order to shed light on their behavior. However, we need first to introduce the caching model.

\section{Caching Model}\label{sec:caching} 
There are definitely certain cases that BSs need to cache the files of users in the uplink and backhaul for further cost savings, latency reduction, etc. For example, imagine a crowded scenario where many users are willing to upload their video recordings to the Internet/network. If uploaded files could be proactively cached at the BSs, that could be beneficial to the network if nearby users have later the interest to download/watch those uploaded videos. Proactive caching at the uplink could alleviate the upload traffic. This can be also a criterion to select which file is of high interest. Specifically, assuming that the nodes/users upload files via uplink,  there is a chance that the user who is uploading the file will have many downloads (i.e., Justin Bieber sharing a video, most likely will be viewed by his followers). Therefore, the file of interest could be inferred based on proactive prediction methods of how the uploading device is "influential", relying on machine learning, social networks, etc. Moreover, although there is no not too much study about the role of caching in the uplink, uplink caching will be very beneficial when the Internet of Things (IoT) devices will be introduced on cellular networks~\cite[Fig. 37]{Indexa}.
 
Let us consider that the network has a \emph{content catalog} of $F$ contents represented by the set of $\mathcal{F} = \{f_1, ..., f_F\}$. User $k$ in $j$th cell at the $n$-th time slot (located at $x_{ljk,n}$) demands a content from a sub-catalog $\mathcal{F}_{l} \subseteq \mathcal{F}$ according to a content popularity distribution $f_{\mathrm{pop}}$. In particular, the BS at the tagged cell  has a content popularity distribution $f_{\mathrm{pop}}$ and is modelled by a \emph{modified} von Mises distribution \cite{Mardia1975}, which is a symmetric circular distribution defined as
	\begin{equation}
		f_{\mathrm{pop}}(f, \mu_j, \gamma_j) = \frac{e^{\gamma_{j} \mathrm{cos}\big( \big(2\pi (f - \mu_j)/F \big)- \pi \big)}}{2\pi J_0 \big(\gamma_j\big)},\label{fpro} 
	\end{equation}
	where $f$ is a point in the support such as $0 \leq f \leq F$, the parameter $\mu_j$ is a measure of location such as $0 \leq \mu_j \leq F$, the parameter $\gamma_j$ is a measure of concentration with $0 \leq \gamma_j \leq +\infty$, and the function $J_0 \big(\gamma_j\big)$ is the zeroth-order Bessel function of the first kind.   The distribution becomes uniform when $\gamma_j = 0$ and highly concentrated on the point $\mu_j$ when $\gamma_j \rightarrow \infty$. The parameters $\mu_j$ and $1/\gamma_j$ are analogous to the mean and variance in Gaussian distribution. In fact, when $\gamma_j \rightarrow \infty$, we obtain the Dirac delta function. The intuition behind such a distribution and modelling is due to the observation that the content catalog is finite $[0, F]$ and the content popularities might be concentrated on specific region, where parameters $\mu_j$ and $1/\gamma_j$ are used for its description.
 Suppose that the $j$-th BS has a storage capacity of $S_{j}$ nats with $S_{j} \leq F$ ($1$ bit = $\ln(2)= 0.693$ nats), and caches files according to the  policy provided below. Henceforth, all the parameters concerning caching are the same across all cells, e.g., we assume that $\mu_j=\mu$ and $\gamma_j=\gamma$. The length of each file in the catalog is $L$ nats, while $T$ expresses its bitrate in nats/s/Hz. Note that the uplink rate of each user has to be equal or higher than the file bitrate $T$, in order to avoid any interruption during its experience.

We assume that we store the most popular files from the catalog in advance offline. Storing most popular files requires perfect knowledge of the content popularity, which might not be possible to be  constructed locally. In order to make the things local/geographical, an interesting caching policy is to store the $S$th closest different files mentioned above.

If the file is of high interest for the BS (if it will be a popular file in the future and is not cached yet), then the BS should cache it, i.e., uplink transmission incurs. If the file is of high interest for the BS (if it will be a popular file in the future) and is already cached, then the BS should do nothing (cache miss). In other words, in such case, cache hit means that the user is in coverage and the file is not included in the BS. Hence, uploading to the BS is meaningful; otherwise, if the user is not in coverage or the file is included at the BS, the file will not be uploaded or it will be uploaded to the core network from the BS.

Notably, the uplink caching  process is dynamic, since the users are likely to have a popular content at any time, which should be uploaded for the better performance of the network. Hence, the proposed model is constantly vital. However, even when all the users have uploaded their contents, they are able to exploit the model and focus on downloading a content that already has been uploaded. The latter scenario is very rare because it is out of chance that at some point all the users will have uploaded their contents. At any time, there will be at least one user that will have some popular content to upload.
\section{Channel Estimation}\label{sec:estimation} 
Let us denote the channel coherence time is $T_{\mathrm{c}}$. We assume that the same time-frequency resources are shared by the users across all cells. Aiming to the characterization of realistic systems, we account for imperfect CSIT due to pilot contamination and channel aging. Let  $\tau$ denote the length of the training period. Obeying to time-division duplex (TDD) design, during the uplink training phase, having duration $\tau$ symbols, the tagged BS obtains the estimated channel. The uplink data transmission phase consists of $T_{\mathrm{c}}-\tau$ symbols. 

Having in mind that the signal from each user is attenuated with distance because of the path-loss, we present the pilot contamination occurred due to the re-use of the pilot sequences during the training phase. 

\subsection{Pilot Contamination}
According to TDD,  estimation of the local CSI takes place during the  uplink training phase, where the same band of frequencies is shared across all cells. Moreover, the $k$th user in each cell is assigned with the same pilot sequence. As a result, pilot contamination occurs and the degradation of the system performance is inevitable. Let the superscript $\mathrm{tr}$ describe the training stage. Furthermore, the scheduled user processes with different pilots, i.e., $\Phi_{k}$, $\Phi_{k^{'}}$, are assumed to be independent. The tagged BS receives a noisy observation  of the channel vector from the  associated sheduled user at time instance $n$. The average power of each transmitted pilot symbol from the sheduled user is      $P_{lk,n}$. Hence, the associated BS  observes the channel $\bh_{llk,n}$ as
\begin{align}
\!\!\!{\by}_{llk,n}^{\mathrm{tr}}&
\!= \! \underbrace{\sqrt{P_{lk,n}}\bh_{llk,n}}_{\minibox[c]{Desired \\signal}}+\!\underbrace{\sum_{j \ne l}\sqrt{P_{jk,n}} \bh_{ljk,n}}_{\minibox[c]{Interference\\part}}\!+\underbrace{ \bN_{llk,n}^{\mathrm{tr}}\bpsi^{\H}_{k}}_{\mbox{noise}},\label{eq:Ypt3}
\end{align}
where the vector $\bpsi_{k}\in \mathbb{C}^{\tau\times 1}$ denotes the training sequence of the $k$th user with $\bpsi_{k}\bpsi_{k}^{\H}$=1, and  $\bN^{\mathrm{tr}}_{llk,n}\in \mathbb{C}^{M\times \tau}$ is the spatially white additive Gaussian
noise matrix with i.i.d. entries distributed as $\mathcal{CN}\left( 0,\frac{\sigma^{2}}{K} \right)$. Note that the channel vectors $\bh_{ljk,n}$, being independent across cells and user distances, are Gaussian distibuted as $\mathcal{CN}\left( \b0,\mathbf{ I}_{M} \right)$. 

The tagged BS estimates $\bh_{llk,n}$ by applying  minimum mean square error (MMSE) estimation to~\eqref{eq:Ypt3}, and by assuming that the tagged BS knows perfectly  the large-scale path-losses $\beta_{ljk}$  for $j\ne l$. Thus, the estimated channel is
\begin{align}\label{estimated1} 
 \hat{\bh}_{llk,n}\Big|_{\| r_{ljk,n}\|}&={\mathrm{E}\!\left[\bh_{llk,n}{\by}_{llk,n}^{\mathrm{tr},\H}  \right]}{\mathrm{E}^{-1}\!\left[{\by}_{llk,n}^{\mathrm{tr}}{\by}_{llk,n}^{\mathrm{tr},\H}\right]}{\by}_{llk,n}^{\mathrm{tr}}\nn\\
 &=\frac{\sqrt{P_{lk,n}}\beta_{llk}}{\sum_{j  }P_{jk,n} \beta_{ljk}+\frac{\sigma^{2}}{K}}{\by}_{llk,n}^{\mathrm{tr}},\!\!
\end{align}
and it is distributed as  $\hat{\bh}_{llk,n}\sim \mathcal{CN}\left( \b 0, \sigma_{\hat{\bh}_{k}}^{2}  \mathbf{I}_{M}\right)$ with variance given as
\begin{align}
 \sigma_{\hat{\bh}_{llk}}^{2}=\frac{{P_{lk,n}}\beta_{llk}^{2}}{\sum_{j}P_{jk,n} \beta_{ljk}+\frac{\sigma^{2}}{K}}.
    \end{align} 
    
Based on the orthogonality
principle of the MMSE estimation, the uncorrelated estimation error vector at time instance $n$ is $\hat{\bee}_{llk,n}={\bh}_{llk,n}-\hat{\bh}_{llk,n}$, being  distributed as $\hat{\bee}_{llk,n}\sim \mathcal{CN}\left( \b 0, \sigma_{\hat{\bee}_{llk}}^{2}  \mathbf{I}_{M}\right)$ with
\begin{align}
\sigma_{\hat{\bee}_{k}}^{2}\!=\!=\beta_{llk}
\left( 1-\frac{{P_{lk,n}}\beta_{llk}}{\sum_{j}P_{jk,n} \beta_{ljk}+\frac{\sigma^{2}}{K}} \right).
\end{align}

\subsection{Channel Aging}
The  relative movement of the $k$th associated user with a comparison to the tagged  BS antennas results in the variation of the channel. Hence, this source of imperfection contributes further to the need for estimation of the channel. Mathematically, we are able to relate the current sample of the channel with its past samples by means of an autoregressive model of order $s$~\cite{Baddour2005}. Herein, for the sake of computational complexity and tractability, we choose an autoregressive model of order $1$, which is a common approach in the literature~\cite{Vu2007}.  Thus, the current channel  at the tagged BS is modeled as
\begin{align}
\bh_{llk,n}  =& \delta \bh_{llk,n-1} + \bee_{llk,n},\label{eq:GaussMarkoModel}
\end{align}
where $\bh_{llk,n-1}$ is the channel in the previous symbol duration, and  $\bee_{llk,n} \in \bbC^{N}$, modelled as a stationary Gaussian random process with i.i.d.~entries and distribution $\cC\cN(\b0,(1-\delta^2)\Id_{M}$, is the uncorrelated channel error because of  the channel variation~\cite{Vu2007}. Regarding $\delta$, it is related to the second-order statistics. Specifically, an appropriate measure for modeling the variation of the channel is its second-order statistics, which can be described by means of the autocorrelation function of the channel. For this role, a widely accepted model is the Jakes model due to its generality and simplicity \cite{Baddour2005}.  The Jakes model describes a propagation medium with two-dimensional isotropic scattering and a monopole antenna at the receiver \cite{Jakes1994}.
In such case, the normalized discrete-time autocorrelation function of the fading channel is expressed by
\begin{align}
r \left( s \right)=J_{0}(2 \pi f_{D}T_{s}|s|),\label{eq:scalarACF}
\end{align}
where $f_{D}$ and $T_{s}$ are the maximum Doppler shift and the channel sampling period. Especially, the maximum Doppler shift $f_{D}$  can be expressed  by means of the relative velocity of the scheduled user $v$, i.e., $f_{D}=\frac{v f_{c}}{c}$, where $c=3\times10^{8}~\nicefrac{m}{s}$ is the speed of light and $f_{c}$ is the carrier frequency. Also, $s$ denotes the delay. Increasing the argument of the Bessel function results in a decrease of the magnitude to zero but with some ripples in the meanwhile. We set $\delta=r[1]$, i.e., we consider a single symbol delay. To this end, we assume that the BS has perfect knowledge of $\delta$.

Remarkably, following the procedure in~\cite{Papazafeiropoulos2015a}, we are able  to write both  pilot contamination and time-variation of the channel as a combination. More concretely, the  channel at time slot $n$ can be written as
\begin{align}   
 \bh_{llk,n}&=\delta {\bh}_{llk,n-1}+\bee_{llk,n}\nonumber\\
&=\delta \hat{\bh}_{llk,n-1}+ \tilde{\bee}_{llk,n},\label{eq:MMSEchannelEstimate}
\end{align}
where $\hat{\bh}_{llk,n-1}$ and $\tilde{\bee}_{llk,n}= \delta \tilde{\bh}_{llk,n-1}+\bee_{llk,n}\sim \mathcal{CN}\left(\b0,\sigma_{\tilde{\bee}_{k}}^{2} {\mathrm{ \bI}}_{M} \right)$ with $\sigma_{\hat{\bee}_{k}}^{2}=\left(1-\delta^{2}\sigma_{\hat{\bh}_{k}}^{2} \right)$
are mutually independent. Hence, the estimated channel of the $k$ scheduled user at time $n$ is provided by $\hat{\bh}_{llk,n}=\delta \hat{\bh}_{llk,n-1}$. Note that in the special case, where $\delta=1$, we obtain a static environment with no user mobility.

\section{Uplink Transmission}\label{sec:uplink} 
In general, the physical representation of a link defines  the probability distribution function (PDF) of this link. Specifically, we face different distributions  depending if we model the desired or the interference part of the received signal. Another example, affecting the  PDF, concerns the choices between multi-antenna and single-antenna BS architecture, and between  single or multi-user transmission.  Notable, herein, we employ the general setting of a large number of antennas deployed by the tagged BS serving multiple users simultaneously. The first step towards the statistical characterization of the powers of the  received signal's parts is to model the uplink transmission.

Thus, accounting for a quasi-static block fading model with frequency-flat fading channels varying for symbol to symbol, the received signal from the associated scheduled user at $x_{llk,n}$ to the tagged BS   during the $n$th time-slot after applying a general decoder $\bq_{llk,n}$ can be expressed as 
\begin{align}
 y_{llk,n}&=\bq_{llk,n}^{\H}\bh_{llk,n} s_{lk,n}+\!\sum_{\left( j,k^{'} \right)\ne \left( l,k \right)}\bq_{llk,n}^{\H}\bh_{ljk^{'},n} s_{jk^{'},n}\nn\\
 &+\bq_{llk,n}^{\H}\bn_{llk}\label{signal},
\end{align}
where  $s_{lk,n}$ is the uplink data symbol of the $k$th scheduled user with $\EE\left[| s_{lk,n}|^{2}\right] =P_{lk,n}$.  The channel vector $\bh_{llk,n} \in \mathbb{C}^{M\times 1}$  denotes the desired  channel vector between the tagged BS  and the associated $k$th scheduled user located at $r_{llk,n}\in \mathbb{R}^{2}$ at  time-instace $n$. Similarly,  $\bh_{ljk^{'},n} \in \mathbb{C}^{M\times 1}$ expresses the interference channel vector from the other users  found at $r_{ljk^{'},n}\in \mathbb{R}^{2}$ far from the typical BS at time-instace $n$. Also, $\bn_{llk} \in \mathbb{C}^{M\times 1}\sim \mathcal{CN}\left( 0,{\sigma}^{2}\Id_{M} \right)$ is the Gaussian thermal noise vector in the uplink data transmission.

Taking into account for the realistic case, where imperfect CSI due to pilot contamination and time-variation of the channel (see~\eqref{eq:MMSEchannelEstimate}), is considered, the received signal by the tagged BS  can be written as
\begin{align}
 y_{llk,n}&=\delta \bq_{llk,n}^{\H}\hat{\bh}_{llk,n-1} s_{lk,n}+\bq_{llk,n}^{\H}\tilde{\bee}_{llk,n} \bs_{k,n}\nn\\&+\!\sum_{\left( j,k^{'} \right)\ne \left( l,k \right)}\bq_{llk,n}^{\H}\bh_{ljk^{'},n} s_{jk^{'},n}+\bq_{llk,n}^{\H}\bn_{k}, \label{filteredsignal}
\end{align}
where we have replaced the current channel by means of~\eqref{eq:MMSEchannelEstimate} with its estimated version\footnote{Note that the replacement concerns only the current desired channel because the interference part is not of direct interest and can be seen as additive noise.}. In~\eqref{filteredsignal}, the first term  expresses the desired signal received by the tagged BS. The second term describes the estimation error effect. Furthermore, the third   term represents the other users  interference, while the last term denotes the post-processed noise.

 {The achievable uplink SINR  from the $k$ scheduled user to the  tagged BS, denoted by $ \mathrm{SINR}_{k}$, is shown in~\eqref{eq: general sum_rate}, where we have treated  the unknown terms at the tagged BS as uncorrelated additive noise. The encoding of the message takes place over many realizations of certain sources of randomness in the model. Specifically,  the expectation operators are taken over the channel estimation error, the small-scale fading  in the interference links, and the thermal noise. Notably, the resultant SINR, provided by~\eqref{DetSINR}, is a random variable because of the randomness accompanying the large-scale path-losses. Thus, we result in the uplink of a multi-user large MIMO  HetNet, where the SINR expression will be investigated by using tools from  stochastic geometry and large random matrix theory.}
%


Herein, we present the derivation of an approximation of the SINR, when maximal-ratio combining (MRC) receivers are employed. As the number of BS antennas grows large, the approximation becomes tighter according to the theory of DEs~\cite{Couillet2011}. Starting from the DE SINR distribution, we derive below the outage probability and the average delivery rate.

Let the tagged BS apply the decoder $\bq_{llk,n}$ to the received signal, being a scaled version of the channel estimate  $\delta \hat{\bh}_{llk,n-1}$. The mathematical expression of the MRC decoder is
\begin{align}
\bq_{llk,n}=\frac{\sum_{j}P_{jk,n} \beta_{ljk}+\frac{\sigma^{2}}{K}}{\sqrt{P_{lk,n}}\beta_{llk}}\delta \hat{\bh}_{llk,n-1},\label{MRC} 
\end{align}
where the scaling of the decoder is applied for the sake of simplicity, but it will not affect the SINR distribution. Also, note that the decoder depends on the estimated channel, obtained during the training phase. Hereafter, we omit the time index from the expressions, while note that the DE expressions are calculated over the channel distributions, i.e., they are conditioned on the BSs positions. 

Let $\overline{\mathrm{SINR}}_{llk}$ be the deterministic SINR, obtained such that $\mathrm{SINR}_{llk}-\overline{\mathrm{SINR}}_{llk}\xrightarrow[ M \rightarrow \infty]{\mbox{a.s.}}0$\footnote{The notation  $\xrightarrow[ M \rightarrow \infty]{\mbox{a.s.}}$ denotes almost sure convergence, while the definition of the term ``deterministic equivalent'' is given by~\cite[Def. 6.1]{Couillet2011}.}.
\begin{proposition}\label{SINR} 
The uplink achievable DE SINR with MRC decoding under the presence of pilot contamination   and channel aging  is given by~\eqref{DetSINR} with $W_{jk}=\beta^{-\epsilon  }_{jjk} \beta_{lj k}$. 
\end{proposition}
\begin{longequation*}[tp]\begin{align}
   \mathrm{SINR}_{llk,n}=\frac{\delta^{2}P_{lk,n} |\bq_{llk,n}^{\H}\hat{\bh}_{llk,n}|^{2}}{P_{lk,n}\EE\left[| \bq_{llk,n}^{\H}\tilde{\bee}_{llk,n}|^{2} \right] +\sum_{\left( j,k^{'} \right)\ne \left( l,k \right)}P_{jk^{'},n} \EE\left[ |\bq_{llk,n}^{\H}\bh_{ljk^{'},n}|^{2}\right] +\|\bq_{llk,n}\|^{2} \sigma^{2}}.\label{eq: general sum_rate} 
\end{align}
\begin{small}
  \begin{align}
   \overline{\mathrm{SINR}}_{llk}= \frac{P_{\mathrm{t}}\delta^{2}\beta^{2\left(1-\epsilon  \right)}_{llk}}{P_{\mathrm{t}}\beta^{\left(1-\epsilon  \right)}_{llk}\!\!\left(\!\sum_{j\ne l }W_{jk}\!+\!\frac{\sigma^{2}}{P_{\mathrm{t}} K}\! \right)\!-\!P_{\mathrm{t}}\delta^{2}\beta^{2\left(1-\epsilon  \right)}_{{llk}}\!+\!\!\!\displaystyle\sum_{\left( j,k^{'} \right)\ne \left( l,k \right)}\!\!\!\!\!\!P_{\mathrm{t}}\Bigg(\!\!\left(\! \frac{\sigma^{2}}{P_{\mathrm{t}} K}\!+\!W_{jk}\right)\!\!W_{jk^{'}} \!+\!
W_{jk}^{2} \!\Bigg)\!+\! \sum_{j}W_{jk}+\frac{\sigma^{2}}{P_{\mathrm{t}} K} }.\label{DetSINR}
\end{align}
\end{small}
\hrule
\end{longequation*}
\begin{proof}
See Appendix~\ref{SINRproof}.
\end{proof}
Notably, the terms $W_{jk}$ correspond to the interference terms from other cells.

\section{Performance Analysis}\label{sec:performance}
The proposed realistic system  depends on several practical
factors, e.g., the pilot contamination and the channel aging as well as caching parameters such as the storage size. As already known, the quality-of-experience constraints specify that the uplink rate of the scheduled user should be equal or higher than the file bitrate $T$ so that the
user does not observe any interruption during its experience. In addition, another impediment is not been taken into account in most cases. It concerns the rate of backhaul, which becomes quite important when the cache misses. 

The quantification and the assessment of the system necessitate the definition of certain metrics, namely the outage probability and the average delivery rate. 
\subsection{Outage probability}\label{outage15}
In this section, we present the uplink outage probability of the associated user in a large antenna MU HetNet with imperfect CSIT due to pilot contamination and channel aging, while caching is employed. The technical derivation is given in Appendix~\ref{Coverageproof}. As a performance metric, the  outage probability is given as the complementary of the success (coverage) probability, expressing the joint probabilities
of the uplink rate exceeding the file bitrate $T$
and the received file missing from the local cache. It is worthwhile to mention that only the BSs connect with the backhaul. In other words, the associated user is not able to upload its content if the  BS already has it. Actually, there is no reason to do it, since the content is stored at the BS and it is BS's task to upload it through its wired link. 
Hence, we have
\begin{align}
\!\!\!\mathbb{P}_{\mathrm{out}}\triangleq 1\!-\!\tilde{\mathbb{P}}\!\left(\overline{\mathrm{SINR}}_{llk} \!>\!\tilde{T}, f\!\not\in\! \Delta_{l} \right)\!,\label{outage}
\end{align}
where $\tilde{T}=e^{T}-1$, $f$ is the received file by the typical BS, and $\Delta_{l}$ is
the local cache of the served  BS at the $l$th cell. Differently to~\cite{Bastug2015}, this definition follows another line of reasoning. In particular, if  the requested file is not in the cache of the served  BS, and if the uplink rate is higher
than the file bitrate $T$, then, the user uploads its content and does not observe any interruption during its communication. Hence, we expect the outage probability to be close to zero. Formally, the outage probability is given by the following theorem.

\begin{theorem}\label{coverage} 
The approximated uplink  outage  probability $\mathbb{P}_{\mathrm{out}}$ in a large MU-MIMO HetNet  with caching attributes, accounting for imperfect CSIT due to pilot contamination and channel aging, is given by
\begin{align}
&\!\!\!\!\!\mathbb{P}_{\mathrm{out}}\approx 1-
\left( 1-\int_{0}^{\frac{S}{L}} f_{\mathrm{pop}}(f, \mu, \gamma) \mathrm{d}f \right)\tilde{\mathbb{P}}\!\left( \overline{\mathrm{SINR}}>\tilde{T}\right)\label{Poutage} 
 \end{align}
with $f_{\mathrm{pop}}(f, \mu, \gamma)$ given by~\eqref{fpro} and the  coverage probability $\tilde{\mathbb{P}}\!\left( \overline{\mathrm{SINR}}>\tilde{T}\right)$ given by~\eqref{Prob1} with $t=\pi \lambda_B x^{2}$. The variable $N$ represents the number of terms used in the calculation, $\eta=N \left( N! \right)^{-\frac{1}{N}}$, while  $D_{\sigma^{2}}=\frac{\sigma^{2}}{KP_{\mathrm{t}}C^{1-\epsilon}\left( \lambda_\mathrm{b}\pi \right)^{\frac{\al\left( 1-\epsilon \right)}{2}}}$, $D_{1}=\frac{2\Gamma^{\al}\left( \frac{\epsilon}{2}+1 \right)}{\al-2}+D_{\sigma^{2}}$, $D_{2}=\left( K-1 \right)D_{1}$, $D_{3}=\left( K-1 \right)\frac{\Gamma^{\al}\left( \epsilon+1 \right)}{\al-1}$, $C_{4}=D_{1}+2C^{1-\epsilon}\left( \pi\lambda_\mathrm{b} \right)^{\frac{ 1-\al\epsilon }{2}}\Gamma^{\al}\left( \frac{\epsilon}{2} +1\right)$, $D_{5}\left( t \right)=\int_{0}^{\infty}\frac{\mathrm{e}^{-u}\mathrm{d}u}{1+D_{2}t^{2 \al\left( 1-\epsilon \right)}u^{-\frac{\al}{2}\left( 1-\epsilon \right)}}$, and $D_{6}=-\frac{1}{P_\mathrm{t}}$.
\end{theorem}
\begin{proof}
See Appendix~\ref{Coverageproof}.
\end{proof}

\subsection{Average Delivery Rate}\label{AchievableRate}
This section presents the derivation of the average delivery rate, defined as
\begin{align}
 \!\!\!\!R\!\triangleq\!\left\{\!\!\!\begin{array}{ll}
T,&\!\!\! \mathrm{if}~\psi\ln\!\left( 1\!+\!\overline{\mathrm{SINR}}_{llk} \right)\!>T\!~\mathrm{and}~f \not \in \Delta_{b_{0}} \\
C\!\left( \lambda_{\mathrm{B}} \right)\!,&\! \!\!\mathrm{if}~\psi\ln\!\left( 1\!+\!\overline{\mathrm{SINR}}_{llk} \right)\!>\!T~\mathrm{and}~f \in \Delta_{b_{0}} \\       
0,&\!\!\!\mathrm{otherwise}\end{array} 
\right.\!\!\label{rate} 
\end{align}
where $C\left( \lambda_{\mathrm{B}} \right)=\frac{C_{1}}{\lambda_{\mathrm{B}}}+C_{2}$ with $C_{1}$ and $C_{2} $ being arbitrary coefficients under the constraint the ceiling of the delivery rate is $C\left( \lambda_{\mathrm{B}} \right)$ with $C\left( \lambda_{\mathrm{B}} \right)<T$. $C\left( \lambda_{\mathrm{B}} \right)$ denotes the backhaul capacity being available to the intermediate nodes. Also,  $\psi=\frac{T_{\mathrm{c}}-\tau}{T_{\mathrm{c}}}$ is the fraction of time expressing the training overhead which occurs during the estimation channel. 
\begin{longequation*}[tp]
  \begin{align}
  \!\!\! \tilde{\mathbb{P}}\!\left( \overline{\mathrm{SINR}}>\tilde{T}\right)&\!\approx\!\sum^{N}_{n=1}\! \binom{N}{n}\!\left( -1 \right)^{n+1}\!\int_{0}^{\infty}\!\!\mathrm{e}^{-t-n\eta \tilde{T}\delta^{-2} \Big( D_{1}t^{\al\left( 1-\epsilon \right)}+D_{3}t^{2\al\left( 1-\epsilon \right)}
 +D_{4}t ^{2\al\left( 1-\epsilon \right)} +D_{6}\Big)  } \!\!\bigg(1- D_{2}t^{2 \al\left( 1-\epsilon \right)} D_{5}\left( t \right) \bigg)^{\!K-1}\mathrm{d}t.\label{Prob1}
\end{align}
\hrule
\end{longequation*}
However,~\eqref{rate} refers to the opposite direction (uplink) and then it includes an interesting and insightful explanation, especially, because the conditions are different. In the case that the uplink rate is higher than the target file rate (bitrate) $T$ and the file is not found in the  BSs, the user uploads at full rate $T$. On the contrary, if the rate is greater than $T$ and the file is already to the local cache, the associated user does not upload its content, but the tagged BS does. The latter constraint relies on the assumption that a high-speed backhaul is not cost-efficient in dense networks.
\begin{theorem}\label{AverageRate} 
The approximated uplink  average delivery rate of the typical BS   in a large  MU-MIMO HetNet  with caching attributes, accounting for imperfect CSIT due to pilot contamination and channel aging, is given by
\begin{align}
\!\!\!\!\!R \!&=\!\left(\!\!T\!-\!\left( C\!\left( \lambda_\mathrm{B} \right)\!-\!T  \right)\!\!\!\int_{0}^{\frac{S}{L}} \!\!\!f_{\mathrm{pop}}(f, \mu, \gamma) \mathrm{d}f \!\right)\!\tilde{\mathbb{P}}\!\left( \overline{\mathrm{SINR}}>{T}\right)\!\!\label{AverageRate1} 
\end{align}
where $f_{\mathrm{pop}}(f, \mu, \gamma)$ is given by~\eqref{fpro}, and the  coverage probability is given by~\eqref{Prob1} with $t=\pi \lambda_B x^{2}$ if we substitute $\tilde{T}$ with $T$.
\end{theorem}
\begin{proof}
See Appendix~\ref{AverageRateproof}.
\end{proof}



 \section{Numerical Results}\label{Numerical} 
In this section, we illustrate the behavior of the analytical expressions concerning the outage probability $\mathbb{P}_{\mathrm{out}}$ and the average delivery rate $R$, which are provided by means of~\eqref{Poutage} and~\eqref{AverageRate1}\footnote{Remarkably,  there is no known result in the literature studying caching in the uplink of a HetNet employing a large number of antennas (massive MIMO). In addition, there is no known reference investigating channel aging in the case of cached-enabled BSs.}. In fact, we investigate the impact of various design parameters such as the BS density $\lambda_{\mathrm{B}}$, the storage size of BSs $S$ in nats, and the target bit-rate $T$ in nats/sec/Hz. Also, the analytical expressions are verified by Monte Carlo simulations.  The simulated curves were obtained by averaging the corresponding expressions over $10^3$ random instances. Actually, the simulated results of the outage probability $\mathbb{P}_{\mathrm{out}}$ and the average delivery user rate $R$ are depicted along with the proposed analytical expressions. Specifically, the bullets correspond to the simulation results, while the ``solid''  lines represent the proposed analytical results by varying their parameters. The discrimination between ``solid'' and ``dot'' lines, where applicable, designates the  results with ``caching'' and ``no caching'', respectively.  The ``no caching'' scenario is obtained by assuming that the content popularity distribution coincides with the Dirac delta function, i.e., $\gamma_j \rightarrow \infty$.
\begin{figure}[!h]
 \begin{center}
 \includegraphics[width=0.95\linewidth]{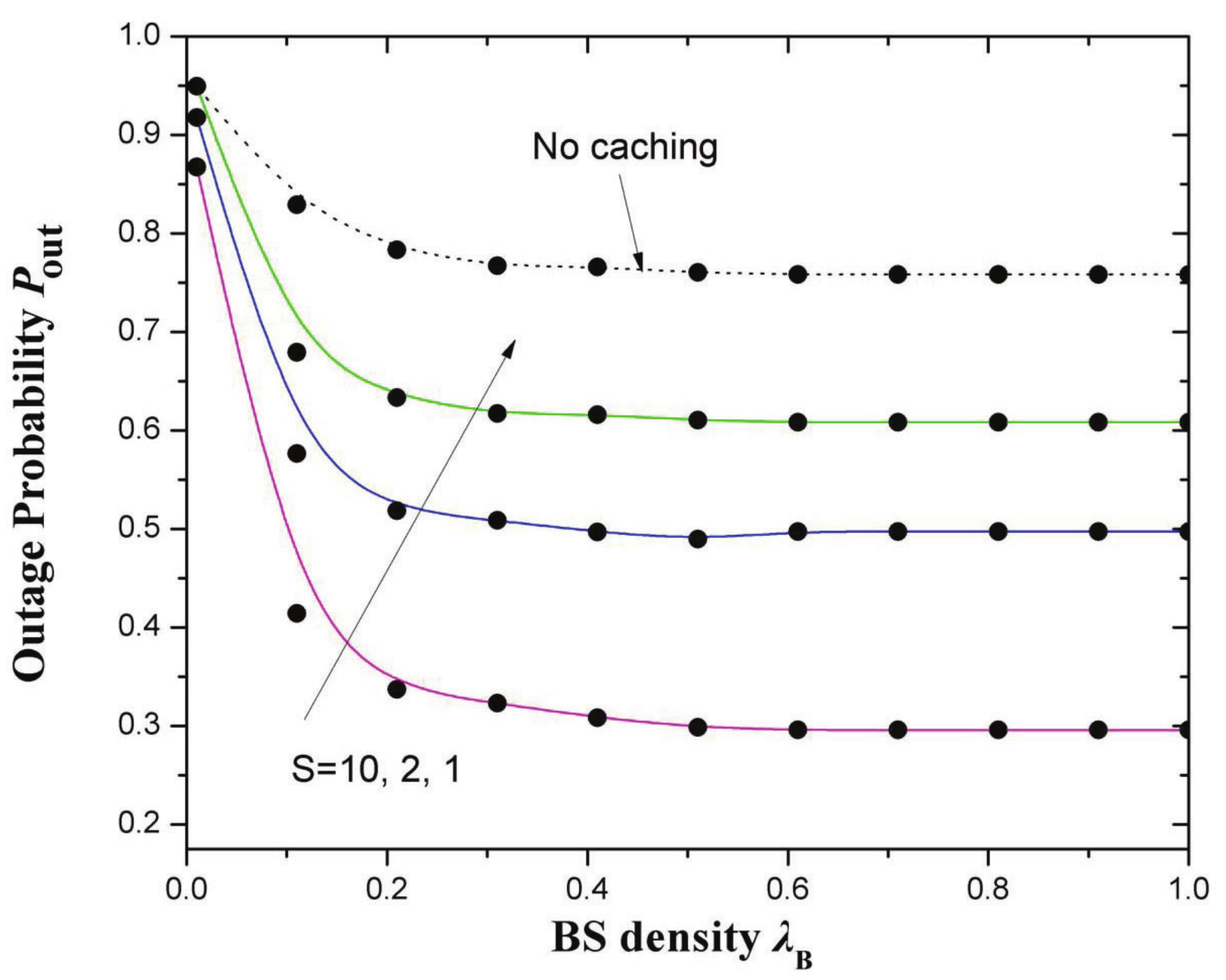}
 \caption{{Outage probability  versus the BS density $\lambda_\mathrm{B}$ for varying  storage size $S$. Solid lines and bullets correspond to the theoretical and simulated results, respectively, while the dotted line refers to the ``No caching'' scenario.}}
 \label{Fig1}
 \end{center}
 \end{figure}
 \begin{figure}[!h]
  \begin{center}
 \includegraphics[width=0.95\linewidth]{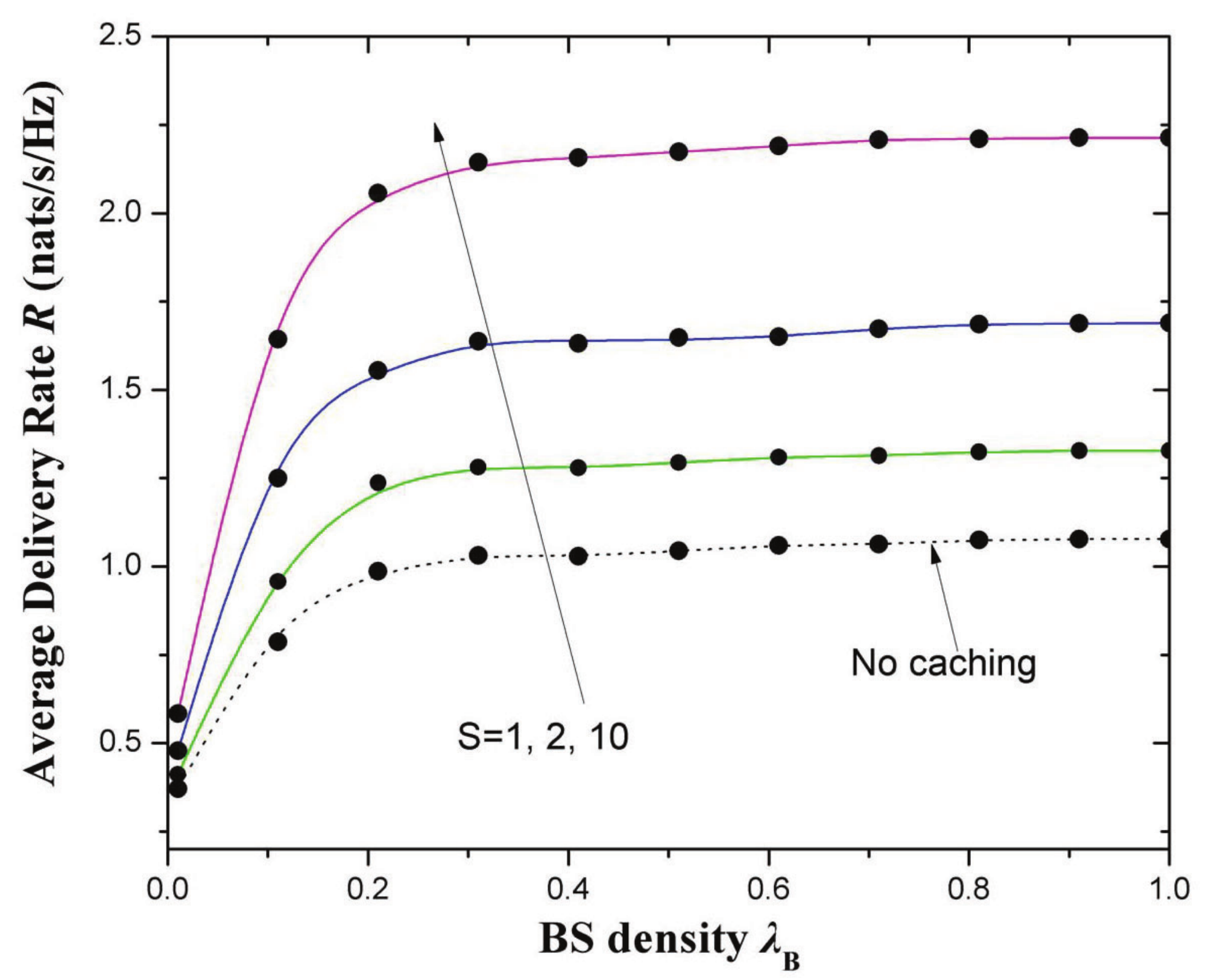}
 \caption{Average delivery rate  versus the BS density $\lambda_\mathrm{B}$ for varying  storage size $S$. Solid lines and bullets correspond to the theoretical and simulated results, respectively, while the dotted line refers to the ``No caching'' scenario.}
 \label{Fig2}
 \end{center}
 \end{figure}

The simulations are conducted by following a specific procedure. Specifically, we choose a sufficiently large area of $5~\mathrm{km}\times 5$ $\mathrm{km}$, where the locations of the BSs are simulated as a realization of a  PPP with given density $\lambda_{\mathrm{B}}=0.2~\mathrm{m}^{-2}$. Next, the users' PPP density is considered to be  $\lambda_{\mathrm{K}}=60\lambda_{\mathrm{B}}$\footnote{Although the analytical expressions rely on the assumption of an infinite plane, the simulation takes place over a finite window.}. The association relies on the minimum path-loss (distance-based) rule, while $K$ users from each cell are randomly scheduled. Hence, we select the strongest user to the tagged BS, found at the origin, as the associated scheduled user at $x_{llk,n}$.   It is worthwhile to mention that the users could employ other schemes to upload their contents to the BSs. For example, they could select the serving BSs rather than the closest BSs. The relevant comparison with other  approaches   regarding the selection  of the appropriate  BSs is   interesting and is left for future work. Furthermore, the  setup includes BSs of $M=25$ number of antennas, while we pick $K=5$ users per BS.   The system under study, embodying a such number of BS antennas,  is considered to describe a massive MIMO model, since the simulations  coincide with the DEs. In other words,  the DEs are tight approximations even for this number of antennas.  Hence, such a number of BS antennas can represent a massive MIMO model. However, this is not a new observation. According to the literature,  Similar observations have been made in the literature even for an $8 \times 8$ system~\cite{Couillet2011,Vaart2000,Wagner2012,Bai2010a}.
The average uplink transmit power for both training and transmission phases is $P_{lk,n}=2~\mathrm{dBW}$, and the bandwidth allocated for each user is $20$ MHz. Also, regarding the rest parameters, we set  $L=1$ nats, $\al=3$, $P_{\mathrm{t}}=1$, $\epsilon=0.7$, $C_{1}=0.005$, $C_{2}=0$, $\delta=0.8$, $\mu=0$, $\gamma=0.5$ unless otherwise stated. Due to limited space, in this work, we do not focus on channel aging, studied in other works such as~\cite{Papazafeiropoulos2017},  but the cynosure is the impact of caching in the uplink.

\subsection{Impact of BS Density}
In Fig.~\ref{Fig1}, we illustrate the behavior of the outage probability $\mathbb{P}_{\mathrm{out}}$ with respect to the BS density $\lambda_{\mathrm{B}}$ for different values of  the storage size $S$. We observe a decrement of the outage probability as the BS density increases. In other words, a denser HetNet provides better coverage. At the same time, an increase of the storage size of the intermediate nodes brings a decrease in an outage, since the users do not have to upload their content to the core network because the  BSs have plenty of space to save the receiver information. 

Regarding the average delivery rate $R$, it increases with the BS density $\lambda_{\mathrm{B}}$ as can be seen in Fig.~\ref{Fig2}. However, it saturates soon due to the increasing intra-cell interference.  Moreover, higher storage size contributes to the increase of the rate because of the traffic load towards the backhaul is alleviated.
\begin{figure}[!h]
\begin{center}
 \includegraphics[width=0.95\linewidth]{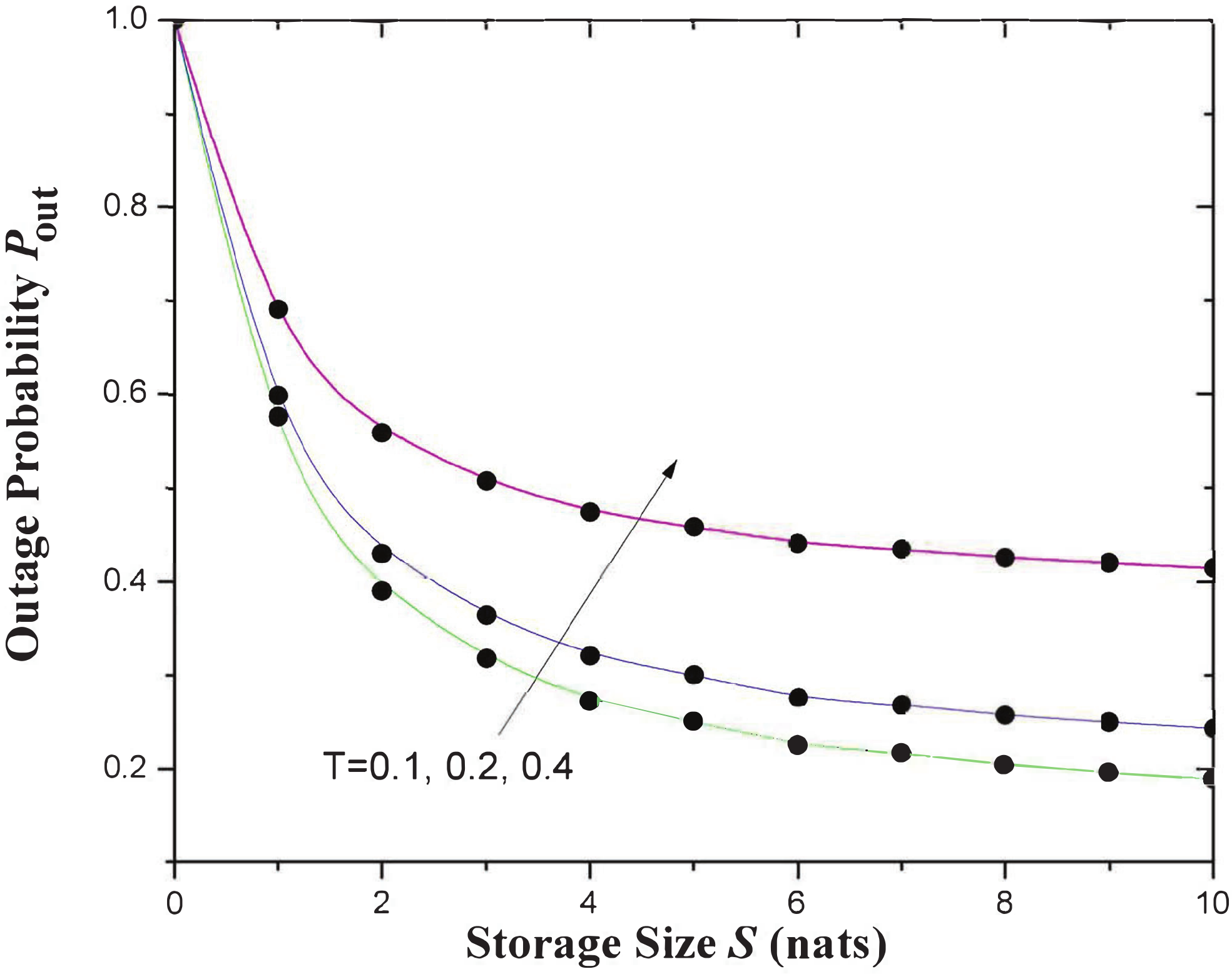}
 \caption{{Outage probability  versus the storage size $S$ for varying  target file bit-rate $T$. Solid lines and bullets correspond to the theoretical and simulated results, respectively.}}
 \label{Fig3}
 \end{center}
 \end{figure}
 \begin{figure}[!h]
  \begin{center}
 \includegraphics[width=0.95\linewidth]{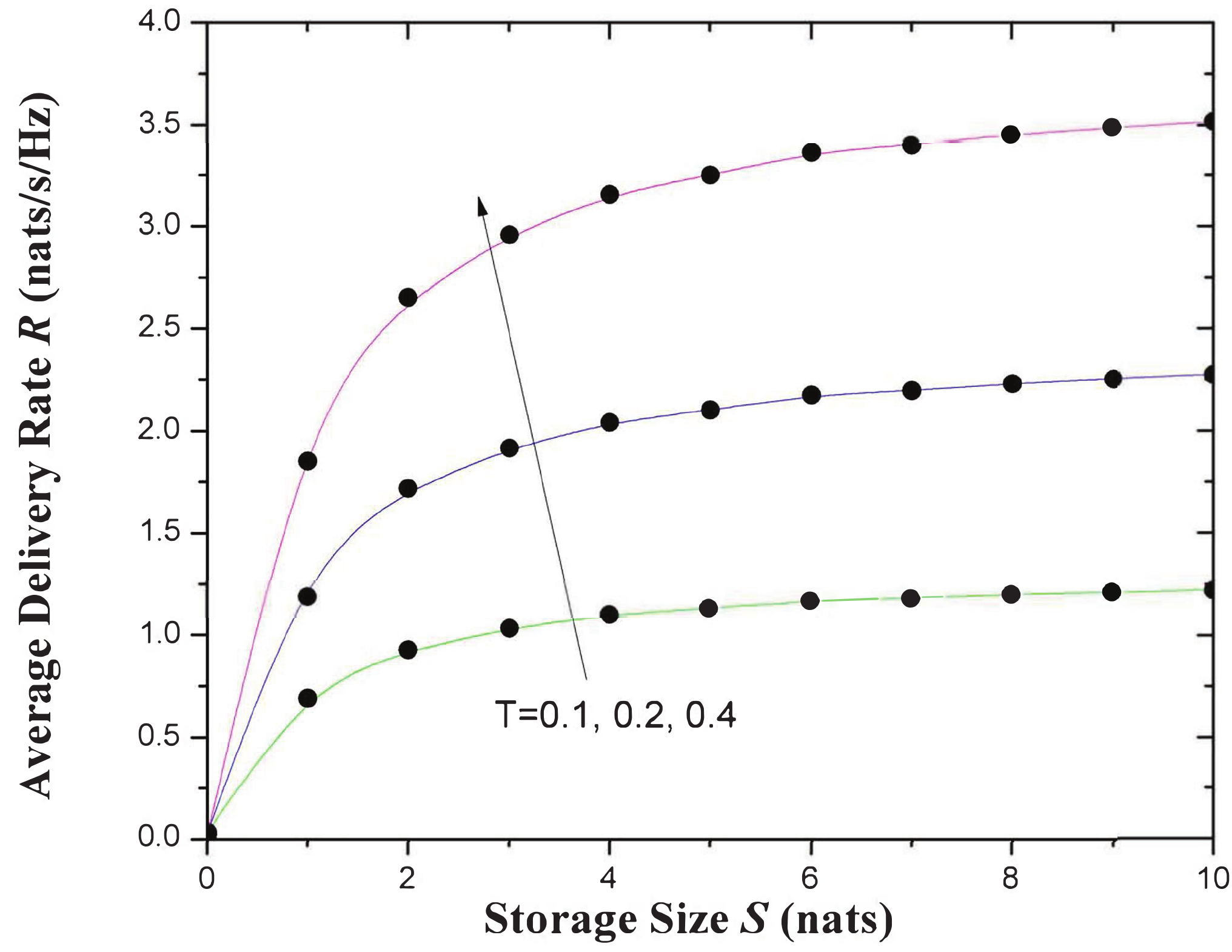}
 \caption{{Average delivery rate  versus the storage size $S$ for varying  target file bit-rate $T$. Solid lines and bullets correspond to the theoretical and simulated results, respectively.}}
 \label{Fig4}
 \end{center}
 \end{figure}
\subsection{Impact of Storage Size}
Fig.~\ref{Fig3} shows the relationship of the outage probability $\mathbb{P}_{\mathrm{out}}$ with the storage size $S$ of the BSs. Notably, the storage capability of networks with caching is one of the most crucial parameters during the design. Obviously, the outage probability increases with the storage size, but decreases  with the target bit-rate.  In other words, the larger the target bit rate is set, the larger the outage probability will be.

In the same direction, in Fig.~\ref{Fig4}, the average delivery rate becomes higher with increasing storage size, but after a value of $S$, further increment is not beneficial, since all users content will be already available to the corresponding  BSs. Especially,  less target rate allows better coverage.
\begin{figure}[!h]
\begin{center}
 \includegraphics[width=0.95\linewidth]{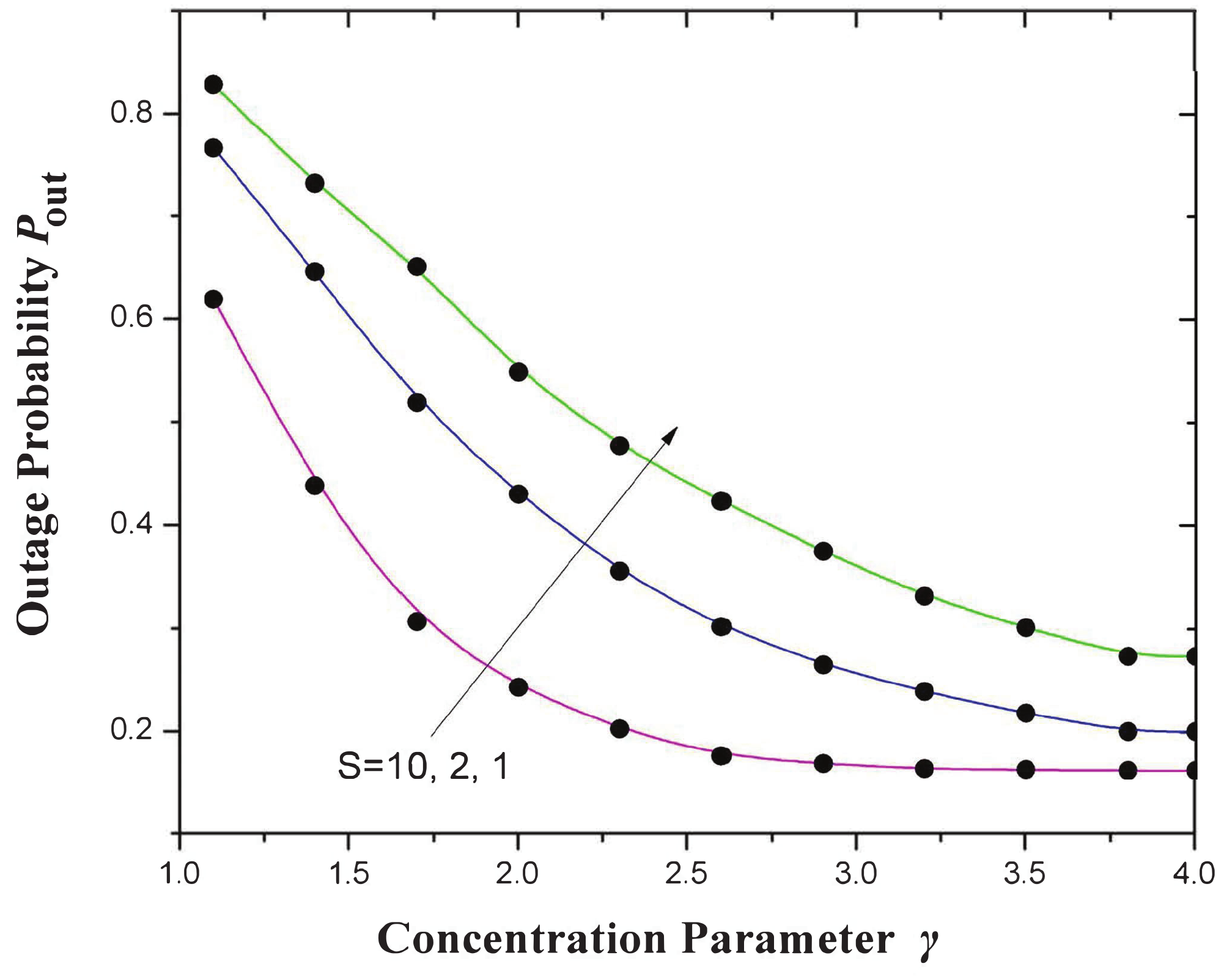}
 \caption{{Outage probability  versus the concentration parameter  $\gamma$ for varying  storage size $S$. Solid lines and bullets correspond to the theoretical and simulated results, respectively.}}
 \label{Fig5}
 \end{center}
 \end{figure}
 \begin{figure}[!h]
  \begin{center}
 \includegraphics[width=0.95\linewidth]{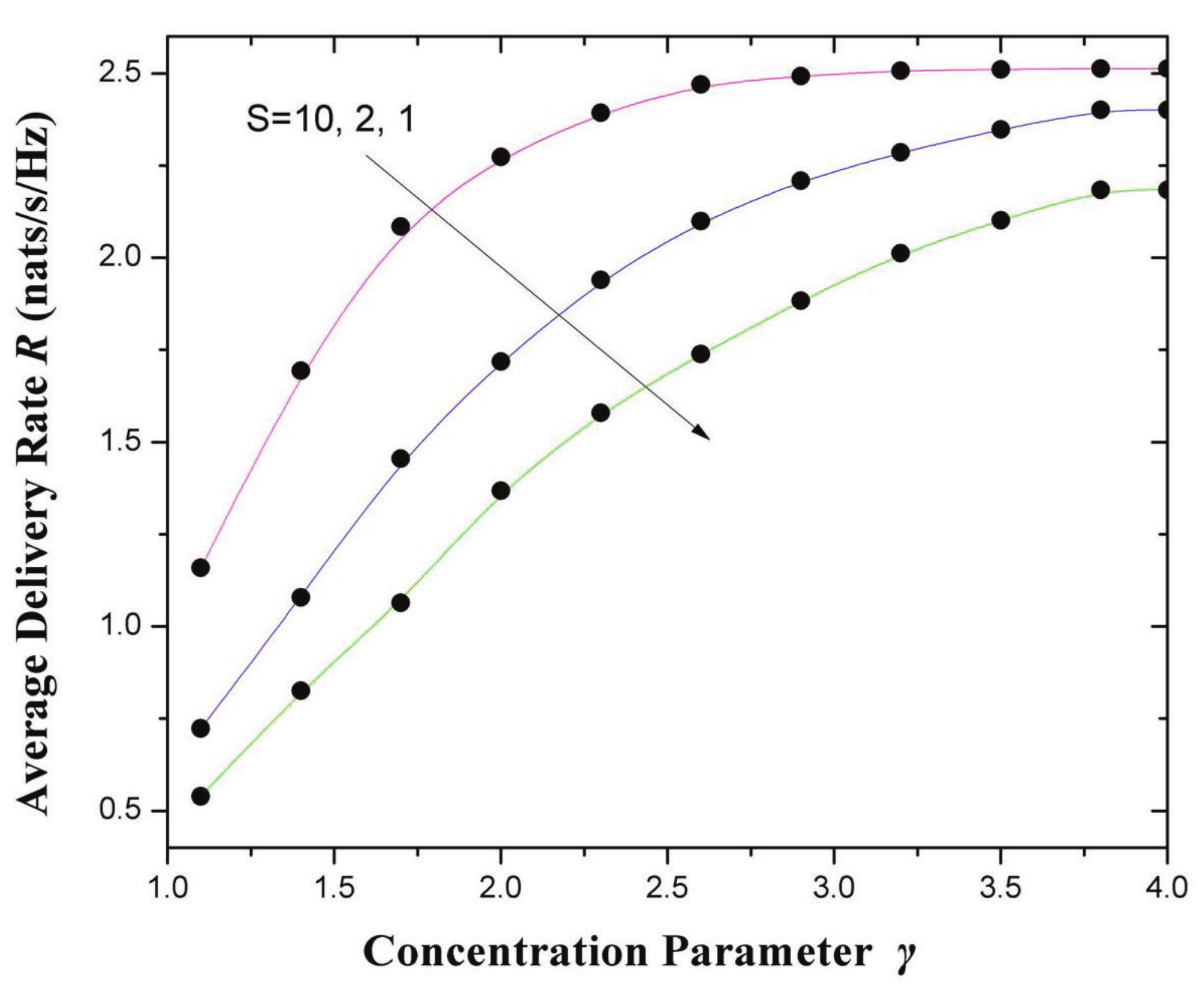}
 \caption{{Average delivery rate  versus the concentration parameter $\gamma$ for varying  storage size $S$. Solid lines and bullets correspond to the theoretical and simulated results, respectively.}}
 \label{Fig6}
 \end{center}
 \end{figure}
\subsection{Impact of Concentration Parameter}
The variation of the concentration parameter, described by $\gamma$, is depicted in Fig.~\ref{Fig5}. Small $\gamma$ means that a high quota of files is already at the intermediate nodes, i.e., many files are popular. Hence, the users do need to upload their files and significant outage is observed. Moreover, higher storage size allows more files to be uploaded in the  BSs. As a result, it is likely that the contents of the users are already at the BSs and the users are inactive since they do not need to upload their contents.

The dependence of the average delivery rate $R$ with the concentration parameter is provided by Fig.~\ref{Fig6}. Specifically, a high concentration parameter means that many files will be uploaded, and thus, the average rate increases.
\begin{figure}[!h]
\begin{center}
 \includegraphics[width=0.95\linewidth]{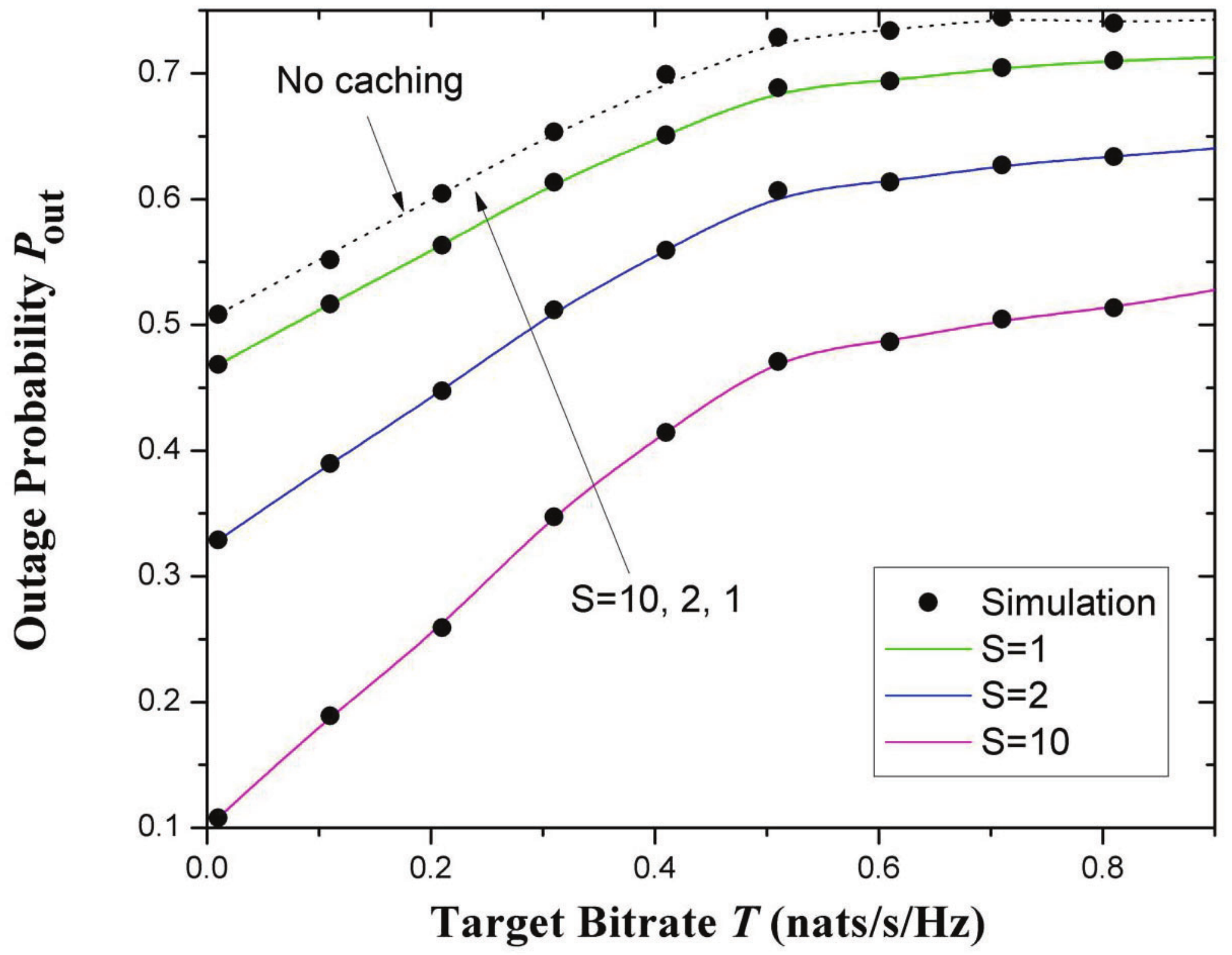}
 \caption{{Outage probability  versus the target bitrate  $T$ for varying  storage size $S$. Solid lines and bullets correspond to the theoretical and simulated results, respectively, while the dotted line refers to the ``No caching'' scenario.}}
 \label{Fig7}
 \end{center}
 \end{figure}
 \begin{figure}[!h]
  \begin{center}
 \includegraphics[width=0.95\linewidth]{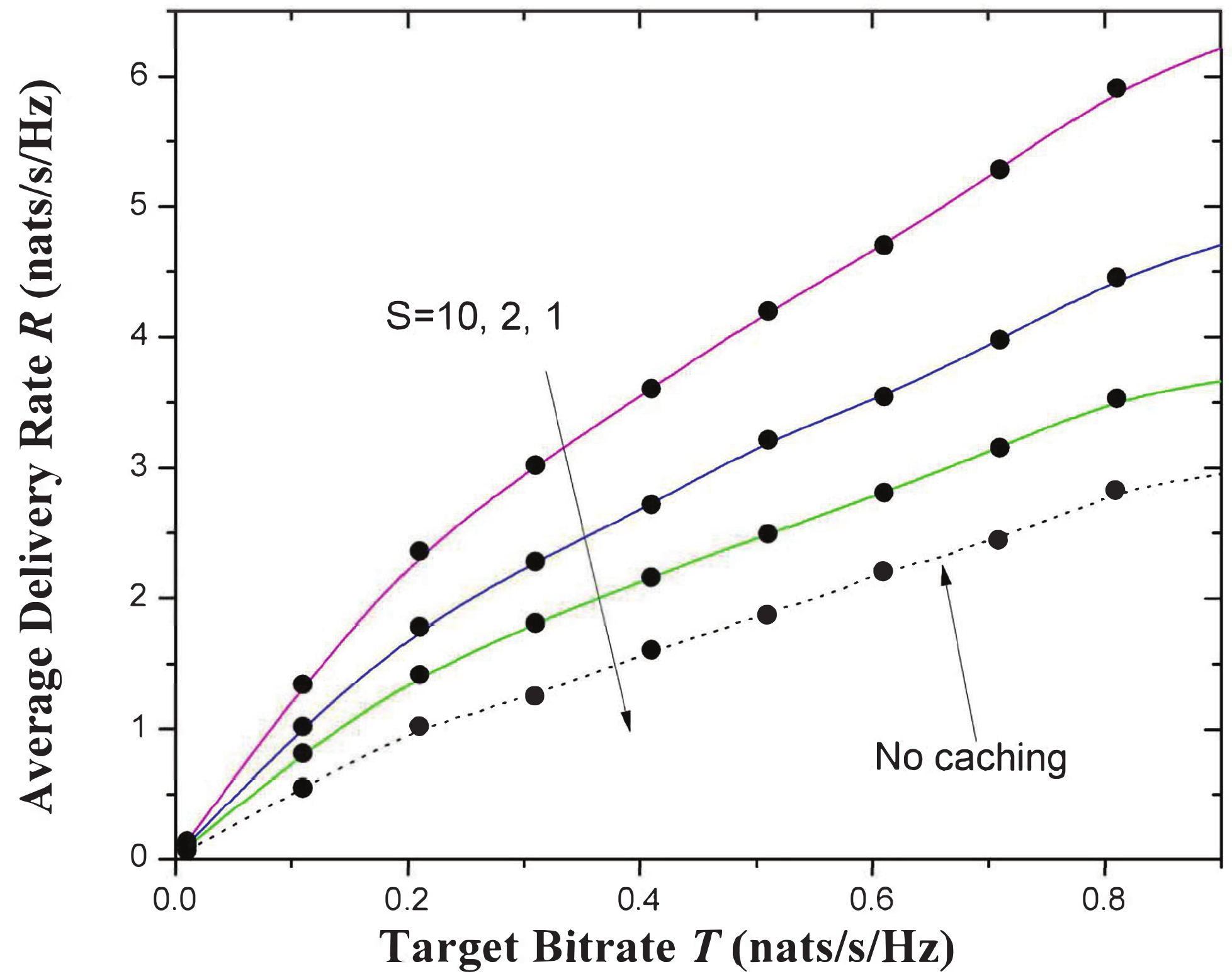}
 \caption{{Average delivery rate  versus the target bitrate  $T$ for varying  storage size $S$. Solid lines and bullets correspond to the theoretical and simulated results, respectively, while the dotted line refers to the ``No caching'' scenario.}}
 \label{Fig8}
 \end{center}
 \end{figure}
\subsection{Impact of Target Bit-Rate}
The target bit-rate $T$ is another critical parameter that should be taken into account during the formation and study of the current architecture. In particular, Fig.~\ref{Fig7} demonstrates the lines of the outage probability $\mathbb{P}_{\mathrm{out}}$ versus the target $T$  for $S=1, 2$, and $10$. Increasing the target rate, the outage probability increases, since less users are served. In addition, the performance is improved with increasing storate size because more content can be saved to the intermidiate nodes without the need to upload it at the backhaul.

In a parallel avenue, Fig.~\ref{Fig8} shows the increase in the performance with bigger storage capacities at the BSs, while it is apparent that a higher target rate results in  a higher average delivery rate.

 \section{Conclusion} \label{Conclusion} 
In this paper,  we introduced the concept of caching in the uplink of a system with stochastically distributed massive MIMO BSs, where users upload their contents to servers through the BS by means of finite-rate backhaul links. In addition to significantly generalizing the state of the art cache-enabled PPP models to the uplink  scenario, we enriched the uplink of the HetNet with the massive MIMO concept. Remarkably, it is the first work, where the caching nodes  have a large number of antennas. Moreover, our approach considered imperfect CSI due to pilot contamination and channel aging. After deriving the DE of the SINR, we provided the outage probability and the average delivery rate. Our main purpose was to focus on fundamental parameters, being relevant to the caching design. Such parameters are the storage size of the serving BS and their target file rate. In particular, we demonstrated that by increasing the storage size, the performance of the system is improved, since the outage probability decreases and the average delivery rate increases. Furthermore, by increasing the target file bitrate, the majority of the users is not served. Hence, the outage probability increases.  Overall, it was shown that the introduction of the notion of caching in the uplink enhances the system performance.

\begin{appendices}

\section{Useful Lemmas}
\begin{lemma}[Alzer's inequality {\cite{Alzer1997}}]\label{AzherInequality}
Assuming that $h$ is a normalized gamma random variable with parameter $N$ and a constant $\gamma>$, then the probability $\mathbb{P}\left( h < \gamma \right)$ can be tightly upper bounded by
\begin{align}
\mathbb{P}\left( h < \gamma \right)<\left( 1-\mathrm{e}^{-\al \gamma} \right)^{N}, 
\end{align}
where $\al=N\left( N! \right)^{-\frac{1}{N}}$.
\end{lemma}

\begin{lemma}[{\cite[Lem. B.26]{Bai2010a}}] \label{lemma:asymptoticLimits}
Let $\bA \in \bbC^{N \times N}$ with uniformly bounded spectral norm (with respect to $N$). Consider $\bx$ and $\by$, where $\bx, \by \in \bbC^{N}$, $\bx \sim \cC\cN(\b0, \bPhi_{x})$ and $\by  \sim \cC\cN(\b0, \bPhi_{y})$, are mutually independent and independent of $\bA$. Then, we have
\begin{align}
&\frac{1}{N}\bx^{\H}\bA\bx - \frac{1}{N}\tr \bA\bPhi_{x}  \xrightarrow[ N \rightarrow \infty]{\mbox{a.s.}} 0 \label{eq:oneVector}\\
&\frac{1}{N}\bx^{\H}\bA\by  \xrightarrow[ N \rightarrow \infty]{\mbox{a.s.}} 0 \label{eq:twoVector}\\
&\EE\!\!\left[\left|\left(\frac{1}{N}\bx^{\H}\bA\bx\right)^{2}\!\! - \!\left(\frac{1}{N}\tr \bA \bPhi_{x} \right)^{2} \right|\right] \!\!\xrightarrow[ N \rightarrow \infty]{\mbox{a.s.}}  0\label{eq:squared}\\
&\frac{1}{N^{2}} |\bx^{\H}\bA\by|^{2} - \frac{1}{N^{2}} \tr \bA \bPhi_{x} \bA^{\H} \bPhi_{y}   \xrightarrow[ N\rightarrow \infty]{\mbox{a.s.}}0. \label{eq:twoVectorGeneral}
\end{align}
\end{lemma}
\section{Proof of Proposition~\ref{SINR}}\label{SINRproof}
First, we divide both the numerator and the denominator of~\eqref{eq: general sum_rate} by$\frac{1}{M^{2}}$. Then, we start with the numerator of the SINR. We insert in~\eqref{eq: general sum_rate} the expression of the MRC decoder given by~\eqref{MRC}. We have\footnote{Let $a_n$ and $b_n$ two infinite sequences. $a_n\asymp b_n$ denotes the equivalence relation $a_n - b_n  \xrightarrow[ N \rightarrow \infty]{\mbox{a.s.}}  0$.}
\begin{align}
 \frac{1}{M^{2}}P_{lk} |\bq_{llk}^{\H}\hat{\bh}_{llk}|^{2}&\stackrel{\text{(a)}}{=}\frac{\left( {P_{lk}}\beta_{llk} \right)^{2}}{\delta^{2}\left( \sum_{j }P_{jk} \beta_{ljk}+\frac{\sigma^{2}}{K} \right)^{2}}\frac{1}{M^{2}}|\bq_{llk}|^{4}\nn\\
&\stackrel{\text{(b)}}{\asymp}\frac{1}{M}P_{\mathrm{t}}^{2}\delta^{4}\beta^{2\left(1-\epsilon  \right)}_{llk},
\end{align}
where (a) follows after substituting the estimated channel given in \eqref{MRC}, while (b) is obtained by means of Lemma~\ref{lemma:asymptoticLimits}, since the covariance of $\bq_{llk}$ is $\delta^{2}\left( \sum_{j}P_{ljk} \beta_{ljk}+\frac{\sigma^{2}}{K} \right)\Id_{M}$. We continue with the first term in the denominator of the SINR including the estimation error. We have
\begin{align}
&\frac{1}{M^{2}}P_{lk }|\bq_{llk }^{\H}\tilde{\bee}_{llk }|^{2}\asymp\frac{1}{M^{2}}P_{lk } \delta^{2} M \beta_{llk }\left( \sum_{j}P_{jk} \beta_{ljk}+\frac{\sigma^{2}}{K} \right) \nn\\
&\times \left( 1-\delta^{2}\frac{{P_{lk }}\beta_{llk }}{\sum_{j}P_{jk } \beta_{ljk }+\frac{\sigma^{2}}{K}} \right)\nn\\
&=P_{\mathrm{t}}^{2}\delta^{2}\beta^{\left(1-\epsilon  \right)}_{llk}\frac{1}{M}\left( \sum_{j\ne l }\beta^{-\epsilon  }_{jjk} \beta_{lj k}+\left( 1-\delta^{2} \right)\beta^{\left( 1-\epsilon \right)  }_{llk} +\frac{\sigma^{2}}{P_{\mathrm{t}} K} \right)\nn
\end{align}
because the covariance of the estimation error is $\beta_{llk }
\left( 1-\delta^{2}\frac{{P_{lk }}\beta_{llk }}{\sum_{j}P_{jk } \beta_{ljk }+\frac{\sigma^{2}}{K}} \right)\Id_{M}$. The next step is to derive the second term in the denominator, which is written as 
\begin{align}
&\frac{1}{M^{2}}\sum_{\left( j,k^{'} \right)\ne \left( l,k \right)} |\sqrt{P_{jk^{'} }}\bq_{llk }^{\H}\bh_{ljk^{'} }|^{2}\nn\\
&= \frac{1}{M^{2}}\!\!\sum_{\left( j,k^{'} \right)\ne \left( l,k \right)}\!\!\!\!\delta^{2}|\sqrt{P_{jk^{'} }}\bn_{llk }^{\mathrm{tr}} \bh_{ljk^{'} }
+\sum_{j^{'} \ne l}\sqrt{P_{j^{'}k }} \bh_{lj^{'}k }\bh_{ljk^{'} }|^{2}\nn\\
&\!\!\!\!\asymp\!\!\!\!\!\!\sum_{\left( j,k^{'} \right)\ne \left( l,k \right)}\!\!\!\!\!\!\!\!\!\delta^{2}\Big( \frac{\sigma^{2}}{M K} P_{\mathrm{t}}\beta^{-\epsilon  }_{jjk^{'}} \beta_{lj k^{'}}
\!+\!\frac{1}{M^{2}}\!\sum_{j^{'} \ne l}\!{P_{j^{'}k }}{P_{jk^{'} }} |\bh_{lj^{'}k }\bh_{ljk^{'} }|^{2}\Big)\label{secondTerm}.
\end{align}
When $k^{'}=k$ and $j^{'}=j$, the second term of the previous expression simplifies to
\begin{align}
\frac{1}{M^{2}}{P_{jk^{'} }}{P_{j^{'} k}}| \bh_{lj^{'}k }\bh_{ljk^{'}}| ^{2}&={P_{jk }^{2}}\frac{1}{M^{2}}\|\bh_{ljk }\|^{4}\nn\\
&\asymp \frac{1}{M} P_{\mathrm{t}}^{2}
\beta_{jjk}^{-2\epsilon}\beta_{ljk}^{2},
\end{align}
where we have used Lemma~\ref{lemma:asymptoticLimits}.
In the case that $j^{'}\ne j$ and $k^{'}=k$ or $k^{'}\ne k$, we have 
\begin{align}
 &\frac{1}{M^{2}}{P_{jk^{'} }}{P_{j^{'}k }} \bh_{lj^{'}k }\bh_{ljk^{'} }\asymp\frac{1}{M}{P_{j^{'}k }}{P_{jk^{'} }}\beta_{lj^{'}k }\beta_{ljk^{'} }\nn\\
 &=\frac{1}{M}P_{\mathrm{t}}^{2}\left( \beta_{j^{'}j^{'}k}\beta_{jjk^{'}} \right)^{-\epsilon
}\beta_{lj^{'}k }\beta_{ljk^{'} }.
\end{align}
Thus,~\eqref{secondTerm} becomes
\begin{align}
&\sum_{\!\!\!\!\left( j,k^{'} \right)\ne \left( l,k \right)}\!\!\!\!\!\!\!\!\delta^{2}\Big( \frac{\sigma^{2}}{M K}M P_{\mathrm{t}}\beta^{-\epsilon  }_{jjk^{'}} \beta_{lj k^{'}}\! \!\sum_{j^{'} \ne l}\!{P_{j^{'}k }}{P_{jk^{'} }} \frac{1}{M^{2}}|\bh_{lj^{'}k }\bh_{ljk^{'} }|^{2}\Big)\nn\\
&\asymp \sum_{\left( j,k^{'} \right)\ne \left( l,k \right)}\!\!\!\!\!\delta^{2}\frac{1}{M}\Big( \frac{\sigma^{2}}{K} P_{\mathrm{t}}\beta^{-\epsilon  }_{jjk^{'}} \beta_{lj k^{'}}\Big.\nn\\ &\Big.+P_{\mathrm{t}}^{2}
\beta_{jjk}^{-2\epsilon}\beta_{ljk}^{2}+P_{\mathrm{t}}^{2}\left( \beta_{j^{'}j^{'}k}\beta_{jjk^{'}} \right)^{-\epsilon
}\beta_{lj^{'}k }\beta_{ljk^{'} }\Big).
\end{align}
Regarding the term that includes the thermal noise, after applying Lemma~\ref{lemma:asymptoticLimits} we have
\begin{align}
\frac{1}{M^{2}}\|\bq_{llk }\|^{2} \sigma^{2}\asymp\frac{1}{M} P_{\mathrm{t}}\delta^{2}\left( \sum_{j}\beta^{-\epsilon  }_{jjk} \beta_{lj k}+\frac{\sigma^{2}}{P_{\mathrm{t}} K} \right).
\end{align}

\section{Proof of Theorem~\ref{coverage}}\label{Coverageproof}
The proof starts by finding first the conditional coverage probability on $x$ as
\begin{align}
 \tilde{\mathbb{P}}\!\left(\overline{\mathrm{SINR}}_{llk} \!>\!\tilde{T}, f\!\not\in\! \Delta_{b_{0}} \right)&=\EE_{x}\!\left[   \tilde{\mathbb{P}}\left(\overline{\mathrm{SINR}}_{llk} \!>\!\tilde{T} \right)|x\right] \nn\\
 &\times\EE_{x}\left[   \mathbb{P}\left(f\!\not\in\! \Delta_{b_{0}}  \right)|x\right].\label{pc} 
\end{align}
Hence, we focus on the derivation of $ \tilde{ \mathbb{P}}\left(\overline{\mathrm{SINR}}_{llk} \!>\!\tilde{T} |x\right)$. Specifically, we propose an approximation for  the out-of-cell interference, described by $W_{jk}$ for all users in each cell, i.e., $k \in \left[1, K \right] $. This approximation will allow the decoupling of the correlated terms and will result in a tractable evaluation of $\mathbb{P}_{c}$. Specifically, by approximating the out-of-cell interference with its mean,  we have~ {\cite{Mungara2015}}
\begin{align}
 \sum_{j\ne l}W_{jk}&=\sum_{j\ne l}\beta^{-\epsilon  }_{jjk} \beta_{lj k}\nn\\
 &\approx \EE \left[ \sum_{j\ne l}\beta^{-\epsilon  }_{jjk} \beta_{lj k} \right]\nn\\
 &=\EE \left[ \sum_{j\ne l}\EE\left[ \beta^{-\epsilon  }_{jjk} \right] \beta_{lj k} \right]\label{coverage1}\\
 &=\lambda_\mathrm{b}C^{-\epsilon} \left( \pi\lambda_\mathrm{b} \right)^{-\frac{\al \epsilon}{2}}\Gamma^{\al}\left( \frac{\epsilon}{2} +1\right)\EE\left[\sum_{j\ne l}\beta_{lj k}  \right]\nn\\ 
 &=\frac{2C^{1-\epsilon}\left( \pi\lambda_\mathrm{b} \right)^{\frac{\al\left( 1-\epsilon \right)}{2}}}{\al-2}\Gamma^{\al}\left( \frac{\epsilon}{2} +1\right),\label{coverage2}
 \end{align}
where in~\eqref{coverage1}, we made the following substitution 
\begin{align}
 \EE\left[ \beta^{-\epsilon  }_{jjk} \right]&=C^{-\epsilon}\left(\EE\left[ r_{ljk }^{\epsilon}\right]   \right)^{\al}\nn\\
 &= \left( \pi\lambda_\mathrm{b} \right)^{-\frac{\al \epsilon}{2}}\Gamma^{\al}\left( \frac{\epsilon}{2} +1\right).\nn
  \end{align}
Moreover,~\eqref{coverage2} is obtained by means of the Campbell's theorem~\cite{Chiu2013a} and the exclusion ball model, described in Sec.~\ref{sec:systemmodel} as
\begin{align}
\EE\left[\sum_{j\ne l}\beta_{lj k}  \right]& =2 \pi \lambda_\mathrm{b} C\int_{R_{e}}^{\infty}x^{-\al}x\mathrm{d}x\nn\\
&=\frac{2C \left( \lambda_\mathrm{b} \right)^{\frac{\al}{2}}\pi\lambda_\mathrm{b} }{\al-2}.
\end{align}
In a similar way, we can write
\begin{align}
 \sum_{j\ne l}W_{jk}^{2}=\frac{2C^{2\left( 1-\epsilon \right)}\left( \pi\lambda_\mathrm{b} \right)^{{\al\left( 1-\epsilon \right)}}}{\al-1}\Gamma^{\al}\left( {\epsilon} +1\right).\label{coverage3}
\end{align}
The approximations~\eqref{coverage2} and \eqref{coverage3} allow the simplification of the SINR, conditioned on $r_{jjk}=x$, by its approximate
\begin{align}
 &\!\!\!\!\!\!\overline{\mathrm{SINR}}\!\approx\! \Big( \! D_{1} \!\left( \pi \lambda_B x^{2} \right)^{\al\left( 1-\epsilon \right)}\!+\!D_{2}\! \left( \pi \lambda_B x^{2} \right)^{2 \al\left( 1-\epsilon \right)}\sum_{k^{'}\ne k}r_{jjk}^{\al\left( 1-\epsilon \right)}\nn\\
 &\!\!\!\!\!\!+D_{3} \left( \pi \lambda_B x^{2} \right)^{2\al\left( 1-\epsilon \right)}+D_{4} \left( \pi \lambda_B x^{2} \right)^{2\al\left( 1-\epsilon \right)} +D_{6}\Big)^{-1},\label{coverage4}
\end{align}
where in~\eqref{coverage4} the variables $D_{i}\in {1,\ldots, 4,6}$ are defined in Theorem~\ref{coverage}. Conditioned on $r_{jjk}=x$, we obtain the approximate distribution of the SINR, given by~\eqref{coverage5},  after substituting its expression from~\eqref{coverage4}.
\begin{longequation*}[tp]
\begin{small}
\begin{align}
& \tilde{\mathbb{P}}\!\left( \overline{\mathrm{SINR}}>\tilde{T}|r_{jjk}=x \right)\nn \\ &\!\approx 
 \tilde{\mathbb{P}}\!\Big(\! 1\!>\tilde{T}\delta^{-2} \Big(\! D_{1}\! \left( \pi \lambda_B x^{2} \right)^{\al\left( 1-\epsilon \right)}\! +\! D_{2}\! \left( \pi \lambda_B x^{2} \right)^{2 \al\left( 1-\epsilon \right)}\!\!\sum_{k^{'}\ne k}\!\!r_{jjk^{'}}^{\al\left( 1-\epsilon \right)}\!+\!D_{3}\! \left( \pi \lambda_B x^{2} \right)^{2\al\left( 1-\epsilon \right)} 
 \!+\!D_{4}\! \left( \pi \lambda_B x^{2} \right)^{2\al\left( 1-\epsilon \right)} \!+\!D_{6}\!\Big) \!\Big) \label{coverage5}\\
& \!\approx 
 \tilde{\mathbb{P}}\!\Big(\! \tilde{g}\!>\tilde{T}\delta^{-2} \Big(\! D_{1} \!\left( \pi \lambda_B x^{2} \right)^{\al\left( 1-\epsilon \right)}\! +\! D_{2} \left( \pi \lambda_B x^{2} \right)^{2 \al\left( 1-\epsilon \right)}\!\!\sum_{k^{'}\ne k}\!\!r_{jjk^{'}}^{\al\left( 1-\epsilon \right)}\!+\!D_{3}\! \left( \pi \lambda_B x^{2} \right)^{2\al\left( 1-\epsilon \right)}
 \!+\!D_{4}\! \left( \pi \lambda_B x^{2} \right)^{2\al\left( 1-\epsilon \right)}\!+\!D_{6} \!\Big) \!\Big) \label{coverage6}\\
&\begin{multlined}[b][0.93\textwidth]
\approx 1-\EE\Bigg[\bigg(1-\exp\bigg( -\eta \tilde{T}\delta^{-2} \Big( D_{1} \left( \pi \lambda_B x^{2} \right)^{\al\left( 1-\epsilon \right)} +D_{2} \left( \pi \lambda_B x^{2} \right)^{2 \al\left( 1-\epsilon \right)}\sum_{k^{'}\ne k}r_{jjk^{'}}^{\al\left( 1-\epsilon \right)}+D_{3} \left( \pi \lambda_B x^{2} \right)^{2\al\left( 1-\epsilon \right)} \\
 ~~~~~+D_{4} \left( \pi \lambda_B x^{2} \right)^{2\al\left( 1-\epsilon \right)} +D_{6}\Big) \bigg) \bigg)^{N} \bigg]
\end{multlined} \label{coverage69}\\
&\begin{multlined}[b][0.93\textwidth]
= \sum^{N}_{n=1} \!\binom{N}{n}\!\left( -1 \right)^{n+1}\EE \!\Bigg[\exp\!\!\bigg(\!\! -n\eta \tilde{T}\delta^{-2} \Big( D_{1} \left( \pi \lambda_B x^{2} \right)^{\al\left( 1-\epsilon \right)} + D_{2} \left( \pi \lambda_B x^{2} \right)^{2 \al\left( 1-\epsilon \right)}\sum_{k^{'}\ne k}y_{k^{'}}^{\al\left( 1-\epsilon \right)}\\+D_{3} \left( \pi \lambda_B x^{2} \right)^{2\al\left( 1-\epsilon \right)}
 +\!D_{4} \left( \pi \lambda_B x^{2} \right)^{2\al\left( 1-\epsilon \right)}\!+\!D_{6} \Big)\!\! \bigg)\!\Bigg]\!\label{coverage7}
\end{multlined}\\
&\begin{multlined}[b][0.93\textwidth]
	= \sum^{N}_{n=1} \binom{N}{n}\left( -1 \right)^{n+1}\exp\bigg( -n\eta \tilde{T}\delta^{-2} \Big( D_{1} \left( \pi \lambda_B x^{2} \right)^{\al\left( 1-\epsilon \right)} + D_{3} \left( \pi \lambda_B x^{2} \right)^{2\al\left( 1-\epsilon \right)}
 +D_{4} \left( \pi \lambda_B x^{2} \right)^{2\al\left( 1-\epsilon \right)} +D_{6}\Big) \bigg)  \\ \times\int_{0}^{\infty}\exp\left( D_{2} \left( \pi \lambda_B x^{2} \right)^{2 \al\left( 1-\epsilon \right)}y^{\al\left( 1-\epsilon \right)-\lambda_B \pi y^{2}} \right)^{K-1}\lambda_B \pi y\mathrm{d}y
 \label{coverage8}
\end{multlined}\\
&\begin{multlined}[b][0.93\textwidth]
  = \sum^{N}_{n=1} \binom{N}{n}\left( -1 \right)^{n+1}\exp\bigg( -n\eta \tilde{T} \delta^{-2}\Big( D_{1} \left( \pi \lambda_B x^{2}  \right)^{\al\left( 1-\epsilon \right)} + D_{3} \left( \pi \lambda_B x^{2} \right)^{2\al\left( 1-\epsilon \right)}
 +D_{4} \left( \pi \lambda_B x^{2} \right)^{2\al\left( 1-\epsilon \right)}+D_{6} \Big) \bigg)  \\ \times \bigg(1- D_{2} \left( \pi \lambda_B x^{2} \right)^{2 \al\left( 1-\epsilon \right)} \!\!\!\int_{0}^{\infty}\!\!\!\!\frac{\mathrm{e}^{-u}\mathrm{d}u}{1+D_{2} \left( \pi \lambda_B x^{2} \right)^{2 \al\left( 1-\epsilon \right)}u^{-\frac{\al}{2}\left( 1-\epsilon \right)}} \bigg)\!^{K\!-\!1}
\label{coverage9}
\end{multlined}
\end{align} 
\end{small}
\hrule
\end{longequation*}
In~\eqref{coverage6}, we have considered the dummy variable $\tilde{g}$, having unit mean and shape parameter $N$, in order to approximate the constant number one. Actually, this approximation becomes tighter as $N$ goes to infinity~\cite{Alzer1997}, since $\lim_{y \to \infty}\frac{y^{y}x^{y-1}\mathrm{e}^{-yx}}{\Gamma\left( y \right)}=\delta\left( x-1 \right)$ with $\delta\left( x \right)$ being Dirac's delta function. In~\eqref{coverage69}, we have applied Alzer's inequality (see Lemma~\ref{AzherInequality}), where $\eta=N \left( N! \right)^{-\frac{1}{N}}$, while afterwards, in~\eqref{coverage7}, we have used the Binomial  theorem. Next,~\eqref{coverage8} is obtained by assuming that $y$ is a Rayleigh random variable. In~\eqref{coverage9}, we set $u=\lambda_B \pi y^{2}$, and we take into account the approximation $\exp\left( -x \right)\approx \frac{1}{1+x}$, in order to make the numerical integration faster. 
Finally, given that $x$ is a Rayleigh random variable, we obtain the uplink SINR distribution as
\begin{align}
 \EE_{x}\!\left[  \tilde{\mathbb{P}}\!\left(\overline{\mathrm{SINR}}_{llk} \!>\!\tilde{T} \right)|x\right] &=\int_{0}^{\infty} \tilde{\mathbb{P}}\!\left( \overline{\mathrm{SINR}}>\tilde{T}|r_{jjk}=x \right)\nn\\
 &\times e^{-\pi \lambda_\mathrm{B}x^{2}}2\pi \lambda_\mathrm{B} x\mathrm{d}x.\label{avpc}
\end{align}
The derivation of the second term of~\eqref{pc} is straightforward. In particular, we assume that all the BSs cache the same amount of files (storage size), while the cache hit probability is independent of the distance $r_{jjk}=x$. Thus, we have
\begin{align}
  \EE_{x} \left[  {\mathbb{P}}\!\left(f\!\not\in\! \Delta_{b_{0}}|x  \right) \right]=1-\int_{0}^{\frac{S}{L}} f_{\mathrm{pop}}(f, \mu, \gamma) \mathrm{d}f.\label{prof} 
\end{align}
Inserting~\eqref{avpc} and~\eqref{prof} into~\eqref{pc} and after some algebraic manipulations, we obtain the coverage probability. The proof is concluded by substituting the coverage probability $\tilde{\mathbb{P}}\!\left(\overline{\mathrm{SINR}}_{llk} \!>\!\tilde{T}, f\!\not\in\! \Delta_{b_{0}} \right)$ in~\eqref{outage}.

\section{Proof of Theorem~\ref{AverageRate}}\label{AverageRateproof}
The average delivery rate $\bar{R}=\EE\left( R \right)$ is obtained by applying the expectation operator over both the fading distribution and the PPP. In particular, we have 
\begin{align}
\bar{R}&=\EE\left[R \right]\nn\\
&=\psi\EE\left[  \tilde{\mathbb{P}}\!\left( \ln\!\left( 1\!+\!\overline{\mathrm{SINR}}_{llk} \right)>T \right)\right.\nn \\
&\times \left. \left( \right.T  {\mathbb{P}}\!\left( f \not \in \Delta_{b_{0}}\right)+C\left( \lambda_\mathrm{B} \right) {\mathbb{P}}\!\left(f  \in \Delta_{b_{0}}  \left. \right)\right)\right]\label{appD1} \\
&=\psi\EE\left[  \tilde{\mathbb{P}}\!\left( \ln\!\left( 1\!+\!\overline{\mathrm{SINR}}_{llk} \right)>T |x\right)\right]\nn\\
&\times \left(\right.\EE\left[ T  {\mathbb{P}}\!\left( f \not \in \Delta_{b_{0}}|x\right)\right]+\EE\left[ C\left( \lambda_\mathrm{B} \right) {\mathbb{P}}\!\left(f  \in \Delta_{b_{0}} |x  \right)\right]\label{appD2}\\
&=\psi\mathcal{I}_{1}\left( \mathcal{I}_{2}+\mathcal{I}_{3} \right),\label{appD3} 
\end{align}
where~\eqref{appD1} follows by applying the definition described by~\eqref{rate}, while~\eqref{appD2} is obtained because of the independence between the different events and the property of linearity of the expectation operator. The derivation of $\bar{R}$ continues with the substitution  of $\mathcal{I}_{1}$, which is basically given by (21) after substituting $\tilde{T}$ with $T$.  Moreover, given that the cache hit probability does not depend on $x$, we have
\begin{align}
 \mathcal{I}_{2}=T\int_{0}^{\frac{S}{L}} f_{\mathrm{pop}}(f, \mu, \gamma) \mathrm{d}f,
\end{align}
while $\mathcal{I}_{3}$ is obtained by multiplying~\eqref{prof} with $C\left( \lambda_\mathrm{B} \right)$. After appropriate substitutions in~\eqref{appD3}, the proof is concluded.
\qed
\end{appendices}
\section*{Acknowledgement}
The authors would like to express their gratitude to Dr. E.~Ba{\c{s}}tuǧ for his help and support in making this work possible.
\bibliographystyle{IEEEtran}

 {\bibliography{mybib}}

\end{document}